\documentclass[12pt]{article}

\usepackage[title]{appendix}
\usepackage{lscape}
\usepackage{amsfonts}
\usepackage{amssymb}
\usepackage{amsmath}
\usepackage{epstopdf}
\usepackage{float}
\usepackage{natbib}
\usepackage[margin=.75in]{geometry}
\usepackage[hang]{caption}
\usepackage{color}
\usepackage[pdftex]{graphicx}
\usepackage{psfrag}
\usepackage{rotating}

\usepackage{bbm}

\newcommand{\qed}{\hfill $\blacksquare$}
\renewcommand{\baselinestretch}{1.2}
\allowdisplaybreaks
\newtheorem{theorem}{Theorem}[section]
\newenvironment{proof}[1][Proof]{\begin{trivlist}
\item[\hskip \labelsep {\bfseries #1.}]}{\end{trivlist}}
\newtheorem{proposition}[theorem]{Proposition}
\newtheorem{lemma}[theorem]{Lemma}

\newtheorem{corollary}[theorem]{Corollary}
\newtheorem{example}[theorem]{Example}

\newcounter{assump}
\renewcommand{\theassump}{(\Alph{assump})}

\newenvironment{assump}{%
	\refstepcounter{assump}\par\vskip\topsep%
	\noindent\textbf{Assumption~\theassump:}
}{\par\vskip\topskip}

\DeclareMathOperator{\argmin}{arg\,min}

\makeatletter
\newcommand*\rel@kern[1]{\kern#1\dimexpr\macc@kerna}
\newcommand*\widebar[1]{%
  \begingroup
  \def\mathaccent##1##2{%
    \rel@kern{0.8}%
    \overline{\rel@kern{-0.8}\macc@nucleus\rel@kern{0.2}}%
    \rel@kern{-0.2}%
  }%
  \macc@depth\@ne
  \let\math@bgroup\@empty \let\math@egroup\macc@set@skewchar
  \mathsurround\z@ \frozen@everymath{\mathgroup\macc@group\relax}%
  \macc@set@skewchar\relax
  \let\mathaccentV\macc@nested@a
  \macc@nested@a\relax111{#1}%
  \endgroup
}
\makeatother

\usepackage{mathtools}
\usepackage{dsfont}

\DeclareMathOperator{\tr}{tr}				 
\DeclareMathOperator{\rank}{rank}	
\DeclareMathOperator{\diag}{diag}
\numberwithin{equation}{section}

\begin{document}
\title{High-dimensional estimation of quadratic variation \\ based on penalized realized variance\thanks{Christensen and Nielsen were supported by the Independent Research Fund Denmark under grant 1028--00030B and 9056--00011B. Podolskij acknowledges funding from the ERC Consolidator Grant 815703 ``STAMFORD: Statistical Methods for High Dimensional Diffusions.''}}

\author{Kim Christensen\thanks{Department of Economics and Business Economics, Aarhus University. (kim@econ.au.dk)}
\and Mikkel Slot Nielsen\thanks{Department of Mathematics, Aarhus University. (mikkel.slot.nielsen@gmail.com)}
\and Mark Podolskij\thanks{Department of Mathematics and Department of Finance, University of Luxembourg. (mark.podolskij@uni.lu)}}

\date{October, 2022}

\maketitle

\begin{abstract}
In this paper, we develop a penalized realized variance (PRV) estimator of the quadratic variation (QV) of a high-dimensional continuous It\^{o} semimartingale. We adapt the principle idea of regularization from linear regression to covariance estimation in a continuous-time high-frequency setting. We show that under a nuclear norm penalization, the PRV is computed by soft-thresholding the eigenvalues of realized variance (RV). It therefore encourages sparsity of singular values or, equivalently, low rank of the solution. We prove our estimator is minimax optimal up to a logarithmic factor. We derive a concentration inequality, which reveals that the rank of PRV is---with a high probability---the number of non-negligible eigenvalues of the QV. Moreover, we also provide the associated non-asymptotic analysis for the spot variance. We suggest an intuitive data-driven subsampling procedure to select the shrinkage parameter. Our theory is supplemented by a simulation study and an empirical application. The PRV detects about three--five factors in the equity market, with a notable rank decrease during times of distress in financial markets. This is consistent with most standard asset pricing models, where a limited amount of systematic factors driving the cross-section of stock returns are perturbed by idiosyncratic errors, rendering the QV---and also RV---of full rank.

\vspace*{0.5cm}

\end{abstract}

\vfill

\thispagestyle{empty}

\pagebreak

\section{Introduction} \label{section:introduction} \setcounter{page}{1}

The covariance matrix of asset returns is a central component, which is required in several contexts in financial economics, such as portfolio composition, pricing of financial instruments, or risk management (e.g., \cite{andersen-bollerslev-diebold-labys:03a}). In the past few decades, estimation of quadratic variation (QV) from high-frequency data has been intensively studied in the field of econometrics. The standard estimator of QV, $\Sigma$, say over a window $[0,1]$, is the realized variance (RV) (e.g., \cite{andersen-bollerslev:98a, barndorff-nielsen-shephard:02a, barndorff-nielsen-shephard:04a,jacod:94a}). Given a decreasing sequence $(\Delta_{n})_{n \geq 1}$ with $\Delta_{n} \to 0$ and equidistant observations $(Y_{i \Delta_{n}})_{i = 0}^{ \lfloor \Delta_{n}^{-1} \rfloor}$ of a $d$-dimensional semimartingale $(Y_{t})_{t \in [0,1]}$, the RV is defined as
\begin{equation} \label{equation:realized-variance}
\widehat{\Sigma}_{n} = \sum_{k=1}^{ \lfloor \Delta_{n}^{-1} \rfloor} ( \Delta_{k}^{n} Y) ( \Delta_{k}^{n} Y)^{ \top},
\end{equation}
where $\Delta_{k}^{n} Y = Y_{k \Delta_{n}} - Y_{(k-1) \Delta_{n}}$ is the increment, and $^{ \top}$ denotes the transpose operator.
The asymptotic properties of $\widehat{\Sigma}_n$ are generally known, when the dimension $d$ is fixed (e.g., \cite{barndorff-nielsen-graversen-jacod-podolskij-shephard:06a, diop-jacod-todorov:13a, heiny-podolskij:20a, jacod:94a, jacod:08a}). In particular, for any semimartingale $\widehat{\Sigma}_n$ is by definition a consistent estimator of $\Sigma$.

When modeling the dynamics of a large financial market over a relatively short time horizon, one is often faced with an ill-posed inference problem, where the number of assets $d$ is comparable to (or even exceeding) the number of high-frequency data $\lfloor \Delta_{n}^{-1} \rfloor$ available to compute $\widehat{ \Sigma}_{n}$. This renders RV inherently inaccurate, and it also implies that $d$ cannot be treated as fixed in the asymptotic analysis. To recover an estimate of $\Sigma$ in such a high-dimensional setting, the covariance matrix has to display a more parsimonious structure.

Currently, the literature on high-dimensional high-frequency data, mostly based on a continuous semimartingale model for the security price process, is rather scarce. \citet*{wang-zou:10a} investigate optimal convergence rates for estimation of QV under sparsity constraints on its entries. \citet*{zheng-li:11a} apply random matrix theory to estimate the spectral eigenvalue distribution of QV in a one-factor diffusion setting. Related work based on random matrix theory and eigenvalue cleaning, or on approximating the high-dimensional system by a low-dimensional factor model, can be found in, e.g., \cite{cai-hu-li-zheng:20a, hautsch-kyj-oomen:12a, lunde-shephard-sheppard:16a}. \citet*{ait-sahalia-xiu:17a} study a high-dimensional factor model, where the eigenvalues of QV are assumed to be dominated by common components (spiked eigenvalue setting), and estimate the number of common factors under sparsity constraints on the idiosyncratic components.

In the context of joint modeling of the cross-section of asset returns, if a factor model is adopted,
one should expect a small number of eigenvalues of $\Sigma$ to be dominant (see also Section \ref{section:simulation}). In other settings, such as when the market is complete and some entries of $Y_{t}$ reflect prices on derivative contracts (spanned by the existing assets and therefore redundant), one should even expect the rank of the volatility to be strictly smaller than the dimension.
In such an instance, the rank of the covariance matrix is smaller than $d$, meaning that the intrinsic dimension of the system can be represented in terms of a driving Brownian motion of lower dimension. In general, identifying a low rank can, for instance, help with the economic interpretation of the model, or it may lighten the computational load related to its estimation or simulation. We refer to several recent studies on identification and estimation of the intrinsic dimension of the driving Brownian motion and eigenvalue analysis of QV in \cite{ait-sahalia-xiu:17a, ait-sahalia-xiu:19b, fissler-podolskij:17a,jacod-lejay-talay:08a,jacod-podolskij:13a, jacod-podolskij:18a}, see also \citet*{kong:17a,kong:20a} in the high-dimensional setting. A related contribution of \citet*{pelger:19a} incorporates a finite-activity jump process in the model.

In this paper, we develop a regularized estimator of $\Sigma$ called the penalized realized variance (PRV). We draw from the literature of shrinkage estimation in linear regression (e.g., \cite{hoerl-kennard:70a, tibshirani:96a}). In particular, we propose the following LASSO-type estimator of $\Sigma$ based on the RV in \eqref{equation:realized-variance}:
\begin{equation} \label{equation:lasso-estimation}
\widehat{\Sigma}_{n}^{ \lambda} = \argmin_{A \in \mathbb{R}^{d \times d}} \bigl( \lVert \widehat{\Sigma}_n - A \rVert_{2}^{2} + \lambda \lVert A \rVert_{1} \bigr).
\end{equation}
Here, $\lambda \geq 0$ is a tuning parameter that controls the degree of penalization, while $\lVert \: \cdot \: \rVert_{1}$ and $\lVert \: \cdot \: \rVert_{2}$ denote the nuclear and Frobenius matrix norm. It has been shown in various contexts that estimators based on nuclear norm regularization possess good statistical properties, such as having optimal rates of convergence and being of low rank. For instance, \citet*{koltchinskii-lounici-tsybakov:11a} study such properties in the trace regression model and, in particular, in the matrix completion problem, where one attempts to fill out missing entries of a partially observed matrix. \citet*{negahban-wainwright:11a} adopt a rather general observation model and, among other things, cover estimation in near low rank multivariate regression models and vector autoregressive processes. Further references in this direction include \cite{argyriou-evgeniou-pontil:08a, bach:08a, candes-recht:09a, recht-fazel-parrilo:10a}. While previous papers on nuclear norm penalization are solely dedicated to discrete-time models, the current work is---to the best of our knowledge---the first to study this problem in a continuous-time It\^{o} setting. This implies a number of subtle technical challenges, since the associated high-frequency data are dependent and potentially non-stationary.

We provide a complete non-asymptotic theory for the PRV estimator in \eqref{equation:lasso-estimation} by deriving bounds for the Frobenius norm of its error. We show that it is minimax optimal up to a logarithmic factor. The derivation relies on advanced results from stochastic analysis combined with a Bernstein-type inequalities presented in Theorems \ref{theorem:bernstein} and  \ref{applyingBernstein}, which constitute some of the major theoretical contributions of the paper. The latter builds upon recent literature on concentration inequalities for matrix-valued random variables, but it is more general as the variables are allowed to be both dependent and unbounded. We discuss the necessary conditions on $d$ and $\Delta_n$ to ensure consistency results.
We further show that in a ``local-to-deficient'' rank setting, where some eigenvalues are possibly decreasing as the dimension of the system increases, the penalized estimator in \eqref{equation:lasso-estimation} identifies the number of non-negligible eigenvalues of $\Sigma$ with high probability.

The estimator depends on a tuning parameter $\lambda$. We exploit the subsampling approach of \citet*{christensen-podolskij-thamrongrat-veliyev:17a} to choose a data-driven amount of shrinkage in the implementation. We also provide a related theoretical analysis for the local volatility, which is more likely to exhibit a lower rank than QV (the latter is only expected to have a near low rank in many settings).

In the empirical analysis, we look at the 30 members of the Dow Jones Industrial Average index to demonstrate that our results are consistent with a standard Fama--French three-factor structure, and one may expect a lower number of factors during crisis.

The paper proceeds as follows. In Section \ref{section:theoretical}, we  establish concentration inequalities for the estimation error $\lVert \widehat{\Sigma}^\lambda_n - \Sigma \rVert_{2}$ and show  minimax optimality up to a logarithmic factor. In Section \ref{IdentifyRank}, we present sufficient conditions to ensure that the rank of $\widehat{\Sigma}^\lambda_n$ coincides with the number of non-negligible eigenvalues of $\Sigma$ with high probability. In Section \ref{section:tuning-parameter}, we show how the tuning parameter $\lambda$ can be selected with a data-driven approach. In Section \ref{section:local-volatility}, the associated non-asymptotic theory for estimation of the instantaneous variance is developed. In Section \ref{section:simulation}, we design a simulation study to demonstrate the ability of our estimator to identify a low-dimensional factor model through $\rank ( \widehat{ \Sigma}_{n}^{ \lambda})$. In Section \ref{section:empirical}, we implement the estimator on empirical high-frequency data from the 30 members of the Dow Jones Industrial Average index. The proofs are included in the appendix.

\subsection*{Notation} \label{section:notation}

This paragraph introduces some notation used throughout the paper. For a vector or a matrix $A$ the transpose of $A$ is denoted by $A^{ \top}$.  We heavily employ the fact that any real matrix $A \in \mathbb{R}^{m \times q}$ admits a singular value decomposition (SVD) of the form:
\begin{equation}
A = \sum_{k=1}^{m \wedge q} s_{k} u_{k} v_{k}^{ \top}
\end{equation}
where $s_{1} \geq \cdots \geq s_{m \wedge q}$ are the singular values of $A$, and $\{u_1,\dots, u_{m\wedge q}\}$ and $\{v_1,\dots, v_{m\wedge q}\}$ are orthonormal vectors. If $m = q$ and $A$ is symmetric and positive semidefinite, its singular values coincide with its eigenvalues and the SVD is the orthogonal decomposition (meaning that one can take $v_{k} = u_{k}$). Sometimes we write $s_{k}(A)$, $v_{k}(A)$ and $u_{k}(A)$ to explicitly indicate the dependence on the matrix $A$.
The rank of $A$ is denoted by $\rank(A)$, whereas
\begin{equation} \label{approxRank}
\rank(A; \varepsilon)= \sum_{k=1}^{m\wedge q} \mathds{1}_{[ \varepsilon, \infty)}(s_{k})
\end{equation}
is the number of singular values of $A$ exceeding a certain threshold $\varepsilon \in [0, \infty)$. Note that $\varepsilon \mapsto \rank (A; \varepsilon)$ is non-increasing on $[0,\infty)$, and c\`{a}dl\`{a}g and piecewise constant on $(0,\infty)$. Furthermore, $\rank (A;0) = m\wedge q$ and $\lim_{ \varepsilon \downarrow 0} \rank (A;\varepsilon) = \rank(A)$.
We denote by $\lVert A \rVert_{p}$ the \textit{Schatten $p$-norm} of $A \in \mathbb{R}^{m \times q}$, i.e.,
\begin{equation} \label{equation:schatten-norm}
\lVert A \rVert_{p} = \Bigl( \sum_{k=1}^{m \wedge q} s_{k}^p\Bigr)^{1/p}
\end{equation}
for $p \in (0,\infty)$. Moreover, $\lVert A \rVert_{ \infty}$ denotes the maximal singular value of $A$, i.e., $\lVert A\rVert_{\infty}= s_{1}$.
In particular, in this paper we work with $p=1$, $p=2$, and $p = \infty$ corresponding to the nuclear, Frobenius, and spectral norm. The Frobenius norm is also induced by the inner product $\langle A,B \rangle = \tr (A^{ \top} B)$ and we have the trace duality
\begin{equation}
\langle A, B \rangle \leq \lVert A \rVert_{1} \lVert B \rVert_\infty.
\end{equation}
For a linear subspace $S \subseteq \mathbb{R}^m$, $S^{\perp}$ denotes the linear space orthogonal to $S$ and $P_{S}$ stands for the projection onto $S$.

Given two sequences $(a_{n})_{n \geq 1}$ and $(b_{n})_{n \geq 1}$ in $(0, \infty)$, we write $a_{n} = o(b_{n})$ if $\lim_{n \rightarrow \infty} a_{n}/b_{n} = 0$ and $a_{n} = O (b_{n})$ if $\limsup_{n \rightarrow \infty} a_{n}/b_{n}$ is bounded. Furthermore, if both $a_{n} = O(b_{n})$ and $b_{n} = O(a_{n})$, we write $a_{n} \asymp b_{n}$.
Finally, to avoid trivial cases we assume throughout the paper that the dimension $d$ of the process $(Y_t)_{t\in [0,1]}$ is at least two.

\section{Theoretical framework} \label{section:theoretical}

We suppose a $d$-dimensional stochastic process $Y_{t} = (Y_{1,t}, \dots, Y_{d,t})^{ \top}$ is defined on a filtered probability space $(\Omega, \mathcal{F}, ( \mathcal{F}_{t})_{t \in [0,1]}, \mathbb{P})$. In our context, $Y_{t}$ is a time $t$ random vector of log-prices of financial assets, which are traded in an arbitrage-free, frictionless market and therefore semimartingale (e.g.,  \cite{delbaen-schachermayer:94a}). In particular, we assume $(Y_t)_{t \in [0,1]}$ is a continuous It\^{o} semimartingale, which can be represented by the stochastic differential equation:
\begin{equation} \label{equation:diffusion}
\mathrm{d}Y_{t} = \mu_{t} \mathrm{d}t + \sigma_{t} \mathrm{d} W_{t}, \qquad t \in [0,1],
\end{equation}
where $(W_{t})_{t \in [0,1]}$ is an adapted $d$-dimensional standard Brownian motion, $( \mu_{t})_{t \in [0,1]}$ is a progressively measurable and locally bounded drift with values in $\mathbb{R}^{d}$, and $(\sigma_{t})_{t \in [0,1]}$ is a predictable and locally bounded stochastic volatility with values in $\mathbb{R}^{d \times d}$. In the above setting the QV (or integrated variance) is $\Sigma = \int_{0}^{1} c_{t} \mathrm{d}t$, where $c_t \equiv \sigma_{t} \sigma_{t}^{ \top}$. We remark that the length $T = 1$ of the estimation window $[0,T]$ is fixed for expositional convenience and is completely without loss of generality, since a general $T$ can always be reduced to the unit interval via a suitable time change.

Note that we exclude a jump component from \eqref{equation:diffusion}. It should be possible to extend our analysis and allow for at least finite-activity jump processes. To do so, we can replace the RV with its truncated version, following the ideas of \citet*{mancini:09a}. We leave the detailed analysis of more general jump processes for future work. Apart from these restrictions, our model is essentially nonparametric, as it allows for almost arbitrary forms of random drift, volatility, and leverage.

\subsection{Upper bounds for the PRV} \label{section:concentration}

The goal of this section is to derive a sharp upper bound on the estimation error $\lVert \widehat{ \Sigma}_{n}^{ \lambda} - \Sigma \rVert_{2}$ that apply with high probability, where the PRV estimator $\widehat{ \Sigma}_{n}^{ \lambda}$ has been defined in \eqref{equation:lasso-estimation}. We first show that the PRV can alternatively be found by soft-thresholding of the eigenvalues in the orthogonal decomposition of $\widehat{ \Sigma}_{n}$.

\begin{proposition} \label{proposition:soft-thresholding}
	Let $\widehat{ \Sigma}_{n} = \sum_{k=1}^{d} s_{k} u_{k}u_{k}^{ \top}$ be the orthogonal decomposition of $\widehat{ \Sigma}_{n}$. Then, the unique solution $\widehat{ \Sigma}_{n}^{ \lambda}$ to \eqref{equation:lasso-estimation} is given by
	\begin{equation}\label{RVthreshold}
	\widehat{ \Sigma}_{n}^{ \lambda} = \sum_{k=1}^{d} \max\bigg\{s_{k} - \frac{ \lambda}{2},0 \bigg\} u_{k} u_{k}^{ \top}.
	\end{equation}
\end{proposition}

\noindent Proposition \ref{proposition:soft-thresholding} is standard (see \cite{koltchinskii-lounici-tsybakov:11a}), but we state it explicitly for completeness.
The interpretation of the PRV representation in \eqref{RVthreshold} is that all ``significant'' eigenvalues of $\widehat{ \Sigma}_{n}$ are shrunk by $\lambda/2$, while the smallest ones are set to 0. Hence, we only retain the principal eigenvalues in the orthogonal decomposition of $\widehat{ \Sigma}_{n}$. The PRV can therefore account for a (near) low rank constraint.
In the next proposition, we present a general non-probabilistic oracle inequality for the performance of $\widehat{ \Sigma}_{n}^{ \lambda}$.
\begin{proposition} \label{boundOn} Assume that $2 \lVert  \widehat{ \Sigma}_{n} - \Sigma \rVert_{ \infty} \leq \lambda$. Then, it holds that
	\begin{equation} \label{boundOn:relation}
	\lVert \widehat{ \Sigma}_{n}^{ \lambda} - \Sigma \rVert_{2}^{2} \leq \inf_{A \in \mathbb{R}^{d \times d}} \Big( \lVert A - \Sigma \rVert_{2}^{2} + \min \big\{2 \lambda \lVert A \rVert_{1}, 3 \lambda^{2} \rank(A) \big\} \Big).
	\end{equation}
	In particular, $\lVert \widehat{ \Sigma}_{n}^{ \lambda} - \Sigma \rVert_{2}^{2} \leq \min \big\{2 \lambda \lVert \Sigma \rVert_{1}, 3 \lambda^{2} \rank(\Sigma) \big\}$.
\end{proposition}
The statement of Proposition \ref{boundOn} is also a general algebraic result that is standard for optimization problems under nuclear penalization (e.g., \cite{koltchinskii-lounici-tsybakov:11a}). It says that an oracle inequality for $\widehat{ \Sigma}_{n}^{ \lambda}$ is available if we can control the stochastic error $\lVert \widehat{ \Sigma}_{n} - \Sigma \rVert_{ \infty}$. The latter is absolute key to our investigation.

In order to assess this error, we need to impose the following assumption on the norms of the drift and volatility processes.

\begin{assump} \label{assumption:A}
	$\sup_{t\in [0,1]} \lVert \mu_{t} \rVert_{2}^{2} \leq \nu_{ \mu}$, $\sup_{t \in [0,1]} \tr(c_{t}) \leq \nu_{c,2}$, and $\sup_{t \in [0,1]} \lVert c_{t} \rVert_{ \infty} \leq \nu_{c, \infty}$ for some constants $\nu_{ \mu}$, $\nu_{c,2}$, and $\nu_{c, \infty}$ with $\nu_{c, \infty} \leq \nu_{c,2}$.
\end{assump}

Mathematically speaking, Assumption \ref{assumption:A} appears a bit strong, since it imposes an almost sure bound on the drift and volatility. We can weaken the requirement on the drift to a suitable moment condition without affecting the rate of $\lVert \widehat{ \Sigma}_{n} - \Sigma \rVert_{ \infty}$ in Theorem \ref{applyingBernstein} below, but as usual the cost of this is more involved expressions. If the volatility does not meet the boundedness condition, which it does not for most stochastic volatility models, one can resort to the localization technique from Section 4.4.1 in \citet*{jacod-protter:12a}. In most financial applications, however, Assumption \ref{assumption:A} is not too stringent if the drift and volatility do not vary strongly over short time intervals, such as a day.

The constants $\nu_{ \mu}$, $\nu_{c,2}$, and $\nu_{c, \infty}$ may depend on $d$, but they should generally be as small as possible to get the best rate (see, e.g., \eqref{equation:optimal-upper-bound}). For example, if the components of $\mu$ and $c$ are uniformly bounded, we readily deduce that
\begin{equation}
\nu_{ \mu} = O(d) \quad \text{and} \quad \nu_{c,2} = O(d).
\end{equation}

To study the concentration of $\lVert \widehat{ \Sigma}_{n} - \Sigma \rVert_{ \infty}$, we present an exponential Bernstein-type inequality, which applies to matrix-valued martingale differences as long as the conditional moments are sufficiently well-behaved. While there are several related concentration inequalities (see, e.g., \cite{minsker:17a, tropp:11a, tropp:12a, tropp:15a} and references therein), existing results require that summands are either independent or bounded. Thus, to be applicable in our setting we needed to modify these.

\begin{theorem} \label{theorem:bernstein}
	Suppose that $(X_{k})_{k=1}^{n}$ is a martingale difference sequence with respect to a filtration $( \mathcal{G}_{k})_{k=0}^{n}$ and with values in the space of symmetric $d \times d$ matrices. Assume that, for some predictable sequence $(C_{k})_{k=1}^{n}$ of random variables and constant $R \in (0, \infty)$,
	\begin{equation} \label{momentControl}
	\big\lVert \mathbb{E} \big[ X_{k}^{p} \mid \mathcal{G}_{k-1} \big] \big\rVert_{ \infty} \leq \frac{p!}{2} R^{p} C_{k}, \quad \text{\upshape{for} } k = 1, \dots, n \text{ \upshape{and} } p = 2, 3, \dots
	\end{equation}
	Then, for all $x,\nu \in (0, \infty)$, it holds that
	\begin{equation}
	\mathbb{P} \Bigg( \Big\lVert \sum_{k=1}^{n} X_{k} \Big\rVert_{ \infty} > x,\, \sum_{k=1}^{n} C_{k} \leq \nu \Bigg) \leq 2d \exp \left (- \frac{x^{2}}{2(xR+ \nu R^{2})} \right).
	\end{equation}
\end{theorem}

Assume that Theorem \ref{theorem:bernstein} applies and that $C_{1}, \dots, C_{n}$ in \eqref{momentControl} can also be chosen to be deterministic. Then, with $\nu = \sum_{k=1}^{n} C_{k}$,
\begin{equation} \label{equation:subexponential}
\mathbb{P} \Bigg( \Big\lVert \sum_{k=1}^{n} X_{k} \Big\rVert_{ \infty} > x \Bigg) \leq 2d \exp \bigg(- \frac{x}{4 \nu R^{2}} \min \{x, \nu R \} \bigg),
\end{equation}
so $\lVert\sum_{k=1}^n X_k \rVert_\infty$ has a sub-exponential tail.
The next theorem derives a concentration inequality for $\lVert \widehat{ \Sigma}_{n} - \Sigma \rVert_{ \infty}$. We remark that, although we have Theorem~\ref{theorem:bernstein} at our disposal, the derivation of Theorem~\ref{applyingBernstein} requires a number of non-standard inequalities in order to prove that $R$ in \eqref{momentControl} can be chosen sufficiently small. For further details, see the discussion in relation to its proof in the supplementary file.

\begin{theorem} \label{applyingBernstein} Suppose that Assumption \ref{assumption:A} holds. Then, there exists an absolute constant $\gamma$ such that
	\begin{equation}
	\mathbb{P} \left( \lVert \widehat{ \Sigma}_{n} - \Sigma \rVert_{ \infty} > x \right) \leq 6d \exp \bigg(- \frac{x}{ \gamma \nu_{c,2}  \nu_{c, \infty} \Delta_{n}} \min \{x , \nu_{c, \infty} \} \bigg),
	\end{equation}
	for all $\displaystyle x \geq \gamma \max \Big\{ \frac{ \nu_{c,2} \nu_{\mu} \Delta_{n}}{ \nu_{c, \infty}}, \sqrt{ \nu_{c,2} \nu_{ \mu} \Delta_{n}}, \nu_{c, \infty} \Delta_{n} \Big\}$.
\end{theorem}
We now combine Proposition \ref{boundOn} and Theorem \ref{applyingBernstein} to deliver a result on the concentration of $\lVert \widehat{ \Sigma}_{n} - \Sigma \rVert_{2}$, which is the main statement of this section.
\begin{theorem} \label{theorem:concentration-prv}
	Suppose that Assumption \ref{assumption:A} holds. For a given $\tau\in (0,\infty)$ consider a regularization parameter $\lambda$ such that
	\begin{equation}
	\lambda \geq  \gamma \max \Bigl\{ \sqrt{ \nu_{c,2} \nu_{c, \infty} \Delta_{n} ( \log(d) + \tau)}, \nu_{c,2} \Delta_{n} ( \log(d) + \tau), \frac{ \nu_{c,2} \nu_{ \mu} \Delta_{n}}{ \nu_{c, \infty}} \Bigr\},
	\end{equation}
	where $\gamma$ is a large absolute constant. Then, with probability at least $1 - e^{- \tau}$,
	\begin{equation} \label{concentrationInequality}
	\lVert \widehat{ \Sigma}_{n}^\lambda - \Sigma \rVert_{2}^{2} \leq \lVert \Sigma - A \rVert_{2}^{2} + \min \{2 \lambda \lVert A \rVert_{1}, 3 \lambda^{2} \rank (A) \},
	\end{equation}
	for all $A \in \mathbb{R}^{d \times d}$.
\end{theorem}
If the drift term is non-dominant, in the sense that $\nu_{ \mu} \leq \nu_{c, \infty}$, it follows that the regularization parameter $\lambda$ should meet
\begin{equation} \label{equation:regularization}
\lambda \geq  \gamma \max \Bigl\{ \sqrt{ \nu_{c,2} \nu_{c, \infty} \Delta_{n} \big( \log(d)+ \tau \big)}, \nu_{c,2} \Delta_{n}( \log(d)+ \tau \big) \Bigr\},
\end{equation}
for an absolute constant $\gamma$.
To get a concentration probability as large as possible without impairing the rate implied by \eqref{concentrationInequality}, we should choose $\tau \asymp \log(d)$. Moreover, the first term of the maximum in \eqref{equation:regularization} is largest for $1/ \Delta_{n} \geq ( \log(d)+ \tau) \nu_{c,2}/ \nu_{c, \infty}$. In light of these observations, the following corollary is an immediate consequence of Theorem \ref{theorem:concentration-prv}, so we exclude the proof.

\begin{corollary} \label{consequenceOfPenConcentration}
	Suppose that Assumption \ref{assumption:A} holds with $\nu_{ \mu} \leq \nu_{c, \infty} \leq \nu_{c,2}$. Assume further that $\Delta_{n}^{-1} \geq 2 \frac{ \nu_{c,2}}{ \nu_{c, \infty}} \log(d)$ and choose the regularization parameter
	\begin{equation} \label{optimalReg}
	\lambda = \gamma \sqrt{ \nu_{c,2} \nu_{c, \infty} \Delta_{n} \log(d)},
	\end{equation}
	where $\gamma$ is a sufficiently large absolute constant. Then, it holds that
	\begin{equation} \label{equation:optimal-upper-bound}
	\lVert \widehat{ \Sigma}_{n}^{ \lambda} - \Sigma \rVert_{2}^{2} \leq 3 \gamma^{2} \nu_{c,2} \nu_{c, \infty} \rank( \Sigma) \Delta_{n} \log(d)
	\end{equation}
	with probability at least $1-d^{-1}$.
\end{corollary}
It follows from \eqref{equation:optimal-upper-bound} that the estimation error of $\widehat{ \Sigma}_{n}^{ \lambda}$ is closely related to the size of $\nu_{c,2}$ and $\nu_{c, \infty}$, which are both determined from the properties of the volatility process. If $c$ is uniformly bounded, $\nu_{c,2} = O(d)$. As emphasized by \citet[Section 1.6.3]{tropp:15a}, we can also often assume that $\lVert c_{s} \rVert_{ \infty}$ and, hence, $\nu_{c, \infty}$ can be chosen independently of $d$. When this is the case, the rate implied by \eqref{equation:optimal-upper-bound} is $\rank( \Sigma) \Delta_{n}d \log(d)$, and since $\rank( \Sigma) \leq d$, Corollary \ref{consequenceOfPenConcentration} implies consistency of $\widehat{ \Sigma}_{n}^{ \lambda}$ when both $\Delta_{n} \rightarrow 0$ and $d \rightarrow \infty$ so long as $d^{2} \log(d) = o( \Delta_{n}^{-1})$. If $\rank( \Sigma)$ can be further bounded by a constant (that does not depend on $d$), the growth condition on $d$ improves to $d \log(d) = o( \Delta_{n}^{-1})$. In contrast, for the estimation error $\lVert \widehat{ \Sigma}_{n} - \Sigma \rVert_{2}^{2} \rightarrow 0$, one cannot expect a better condition than $d^{2} = o( \Delta_{n}^{-1})$, since this estimation error corresponds to a sum of $d^{2}$ squared error terms, each of the order $\Delta_{n}$.

\subsection{Minimax optimality}

We denote by $\mathcal{C}_{r}$, for a given non-zero integer $r \leq d$, the subclass of $d \times d$ symmetric positive semidefinite matrices $\mathcal{S}_{+}^{d}$ whose \emph{effective rank} is bounded by $r$:
\begin{equation} \label{optClass}
\mathcal{C}_{r} = \bigl\{A \in \mathcal{S}_{+}^{d} : r_e (A) \leq r \bigr\},
\end{equation}
where
\begin{equation}\label{effectiveRank}
r_{e} (A) \equiv \frac{\tr(A)}{\lVert A \rVert_{ \infty}}.
\end{equation}
Compared to the rank, the effective rank is a more stable measure for the intrinsic dimension of a covariance matrix (see, e.g., \cite[Remark 5.53]{vershynin:10a}).
In the following we argue that $\widehat{ \Sigma}_{n}^{ \lambda}$, with $\lambda$ of the form \eqref{optimalReg}, is a minimax optimal estimator of $\Sigma$ in the parameters $\nu_{c,2}$, $\nu_{c, \infty}$, and $\rank( \Sigma)$ over the parametric class of continuous It\^{o} processes generated by \eqref{equation:diffusion} with no drift $\mu_{s} \equiv 0$ and a constant volatility $\sigma_{s} \equiv \sqrt{A}$ for $A \in \mathcal{C}_{r}$. To this end, denote by $\mathbb{P}_{A}$ a probability measure for which $(Y_{t})_{t \in [0,1]}$ is defined as in \eqref{equation:diffusion} with no drift and constant volatility $c_{s} \equiv A$. In this setting, we can choose $\nu_{c,2} = \tr(A)$ and $\nu_{c,\infty} = \lVert A \rVert_{\infty}$, which by Corollary \ref{consequenceOfPenConcentration} means that
\begin{equation} \label{optEstimator}
\sup_{A \in \mathcal{C}_{r}} \mathbb{P}_{A} \Big( \lVert \widehat{ \Sigma}_{n}^{ \lambda} - \Sigma \rVert_{2}^{2} > \overline{ \gamma} \lVert A \rVert_{ \infty}^{2} r_{e}(A) \rank(A) \Delta_{n} \log(d) \Big) \leq \frac{1}{d},
\end{equation}
for an absolute constant $\overline{ \gamma}$ and $\Delta_{n}^{-1} \geq 2r \log(d)$.

Now, since the log-price increments $\Delta_{1}^{n} Y, \dots, \Delta_{ \lfloor \Delta_{n}^{-1} \rfloor}^{n} Y$ are i.i.d. Gaussian random vectors under $\mathbb{P}_{A}$, we can exploit existing results from the literature to assess the performance of $\widehat{ \Sigma}_{n}^{ \lambda}$. Hence, the following is effectively an immediate consequence of Theorem 2 in \citet*{lounici:14a}. It shows that, up to a logarithmic factor, no estimator can do better than \eqref{optEstimator}.

\begin{theorem} \label{lowerBound} Let $\mathcal{C}_r$ be given as in \eqref{optClass} and suppose that $\lfloor \Delta_n^{-1}\rfloor \geq r^2$. Then, there exist absolute constants $\beta \in (0,1)$ and $\underline{\gamma}\in (0,\infty)$ such that
	\begin{equation}
	\inf_{\widehat{\Sigma}}\sup_{A\in \mathcal{C}_r} \mathbb{P}_A \bigg( \lVert \widehat{\Sigma} - \Sigma \rVert_{2}^{2} > \underline{ \gamma} \frac{ \lVert A \rVert_{ \infty}^{2} r_e(A) \rank(A)}{ \lfloor \Delta_{n}^{-1} \rfloor} \bigg) \geq \beta,
	\end{equation}
	where the infimum runs over all estimators $\widehat{\Sigma}$ of $\Sigma$ based on $( \Delta_{1}^{n}Y, \dots, \Delta_{ \lfloor \Delta_{n}^{-1} \rfloor}^{n}Y)$.
\end{theorem}

\subsection{Bound on the rank of PRV} \label{IdentifyRank}

In this section, we study the rank of $\widehat{ \Sigma}_{n}^{ \lambda}$ relative to the number of non-negligible eigenvalues and, in particular, the true rank of $\Sigma$. In line with Section \ref{section:concentration}, we begin by stating a general non-probabilistic inequality in Theorem \ref{theorem:rank-relation}. In the formulation of this result, we recall that $\rank(A; \varepsilon)$ is the number of singular values of $A$ exceeding $\varepsilon \in [0, \infty)$.

\begin{theorem} \label{theorem:rank-relation} Suppose that $2 \lVert \widehat{ \Sigma}_{n} - \Sigma \rVert_{ \infty} \leq \bar{ \lambda}$ for some $\bar{ \lambda} \in (0, \infty)$. Then
	\begin{equation} \label{rank1}
	\rank( \Sigma; \lambda) \leq \rank( \widehat{ \Sigma}_{n}^{ \lambda}) \leq \rank \bigg( \Sigma; \frac{1}{2}( \lambda - \bar{ \lambda}) \bigg),
	\end{equation}
	for any $\lambda \in ( \bar{ \lambda}, \infty)$. In particular, if $\lambda \in ( \bar{ \lambda}, s]$, where $s$ is the smallest non-zero eigenvalue of $\Sigma$, then $\rank( \widehat{ \Sigma}_{n}^{ \lambda}) = \rank( \Sigma)$.
\end{theorem}
With this result in hand we can rely on the exponential inequality for the quantity $\lVert \widehat{ \Sigma}_{n} - \Sigma \rVert_{ \infty}$ established in Theorem \ref{applyingBernstein} to show that, with high probability and in addition to converging to $\Sigma$ at a fast rate, $\widehat{ \Sigma}_{n}^{ \lambda}$ automatically has the rank of $\Sigma$ when we neglect eigenvalues of sufficiently small order. In particular, if $\Sigma$ has full rank, but many of its eigenvalues are close to zero, $\widehat{ \Sigma}_{n}^{ \lambda}$ is of low rank and reflects the number of important components (or factors) in $\Sigma$.
\begin{corollary} \label{rankConcentration}
	Suppose that Assumption \ref{assumption:A} holds with $\nu_{ \mu} \leq \nu_{c, \infty} \leq \nu_{c,2}$ and fix $\delta \in (0,1/2)$. Assume further that $\Delta_{n}^{-1} \geq 2 \frac{ \nu_{c,2}}{ \nu_{c, \infty}} \log(d)$ and, for a sufficiently large constant $\gamma$ depending only on $\delta$, choose the regularization parameter $\lambda$ as follows:
	\begin{equation} \label{rankRegularization}
	\lambda = \gamma \sqrt{ \nu_{c,2} \nu_{c, \infty} \Delta_{n} \log(d)}.
	\end{equation}
	Then, with probability at least $1-d^{-1}$, it holds that
	\begin{equation} \label{lowRankconcentration}
	\lVert \widehat{ \Sigma}_{n}^{ \lambda} - \Sigma \rVert_{2}^{2} \leq 3 \gamma^{2} \nu_{c,2} \nu_{c, \infty} \rank( \Sigma) \Delta_{n} \log(d),
	\end{equation}
	and
	\begin{equation} \label{rankCor}
	\rank( \Sigma; \lambda) \leq \rank( \widehat{ \Sigma}_{n}^{ \lambda}) \leq \rank( \Sigma; \delta \lambda).
	\end{equation}
	If, in addition, $\lambda \leq s$, where $s$ is the smallest non-zero eigenvalue of $\Sigma$, then both $\rank( \widehat{ \Sigma}_{n}^{ \lambda}) = \rank( \Sigma)$ and \eqref{lowRankconcentration} hold with probability at least $1-d^{-1}$.
\end{corollary}
In the setting of Corollary \ref{rankConcentration}, it follows that with high probability an eigenvalue $s$ of $\Sigma$ affects $\rank( \widehat{ \Sigma}_{n}^{ \lambda})$ for large $n$ if $\lambda = o(s)$, while it does not if $s = o( \lambda)$. The first condition says that, relative to the level of shrinkage, an eigenvalue is significant (or non-negligible), whereas the second says the opposite. Hence, the notion of negligibility depends on $\lambda$, which in turn depends on the model through the constants $\nu_{c,2}$ and $\nu_{c, \infty}$. However, we know that for $\widehat{ \Sigma}_{n}^{ \lambda}$ to be a consistent estimator of $\Sigma$, a necessary condition is that $\lambda \rightarrow 0$ as $n \rightarrow \infty$, which implies that for an eigenvalue of $\Sigma$ to be negligible, it must tend to zero as $n$ increases. In particular, if $d$ is fixed, $\rank( \widehat{ \Sigma}_{n}^{ \lambda}) = \rank( \Sigma)$ with a probability tending to one as $n \to \infty$.

The following example illustrates a stylized setting, where many eigenvalues of $\Sigma$ are negligible.

\begin{example}
	Let $r\in \mathbb{N}$ be an absolute constant (i.e., independent of $d$ and $n$). In line with the factor model studied in \cite{ait-sahalia-xiu:17a}, suppose that $c_{t}$ is of the form
	\begin{equation}\label{specificExample}
	c_{t} = \beta e_{t} \beta^{ \top} + g_{t},
	\end{equation}
	where $e_{t} \in \mathbb{R}^{r \times r}$ and $g_{t} \in \mathbb{R}^{d \times d}$ are predictable, symmetric, and positive definite processes, which are uniformly bounded such that $\lVert e_{t} \rVert_{ \infty} \leq C_{e}$ and $\lVert g_{t} \rVert_{ \infty} \leq C_{g}$ for some constants $C_{e}, C_{g} = O(1)$. The matrix $\beta \in \mathbb{R}^{d \times r}$ of factor loadings is deterministic and constant in time; a common assumption in the literature, which is also supported by \citet*{reiss-todorov-tauchen:15a}. The form \eqref{specificExample} of $c_{t}$ can be motivated by a standard multi-factor model, where the total risk associated with any of the $d$ assets can be decomposed into a systematic and an idiosyncratic component. The systematic component $\beta e_{t} \beta^{ \top}$ is composed of loadings on the $r$ underlying priced common factors in the economy, $\beta$, and the risk of those factors, $e_{t}$. The idiosyncratic component $g_{t}$ is asset-specific and can therefore be reduced by diversification.
	
	Suppose that, given an initial set of $d$ individual securities, we start forming a new set of (orthogonalized) portfolios by taking linear combinations of the original assets, as prescribed by the APT model of \citet*{ross:76a}, and a standard way to implement factor models in practice when constructing sorting portfolios. Then, these new assets are generally diversified enough to assume that $C_{g} = O(d^{-1})$. To identify the number of driving factors, $r$, we assume that $\lVert d^{- \gamma} \beta^{ \top} \beta - I_{r} \rVert_{ \infty}  = o(1)$ (as $d \rightarrow \infty$) and $s_{r}(e_{t}) \geq \varepsilon$ for some absolute constants $\gamma \in [0,1]$ and $\varepsilon \in (0, \infty)$. This corresponds to Assumption 4 in \citet*{ait-sahalia-xiu:17a}, but it is slightly more general as the assets are assumed to be diversified portfolios.
	
	Now, given a suitable drift process $( \mu_{t})_{t \in [0,1]}$, Corollary \ref{rankConcentration} applies with regularization parameter $\lambda = O \big(d^{ \gamma} \sqrt{ \Delta_{n} \log(d)} \big)$. In particular, due to the fact $0 \leq s_{k} ( \Sigma) - s_{k} ( \beta E \beta^{ \top}) \leq \lVert \Gamma \rVert_{ \infty}$ with $E = \int_0^1 e_t \mathrm{d} t$ and $\Gamma = \int_0^1 g_t \mathrm{d} t$, it follows from \eqref{rankCor} that
	\begin{equation}\label{rankExample}
	\rank ( \beta E \beta^{ \top}; \lambda) \leq \rank( \widehat{ \Sigma}_{n}^{ \lambda}) \leq \rank (\beta E \beta^\top ; \delta \lambda-\lVert \Gamma \rVert_\infty)
	\end{equation}
	with probability at least $1 - d^{-1}$. By assumption $\vert s_{r}( \beta E \beta^{ \top}) - d^{ \gamma} s_{r}(E) \vert = o(d^{ \gamma})$ and $s_{r}(E) \geq \varepsilon$ (the former follows from arguments as in the proof of Theorem~1 in \cite{ait-sahalia-xiu:17a}), and hence $d^{ \gamma} = O(s_{r}( \beta E \beta^{ \top}))$. Combining this with \eqref{rankExample}, we conclude that $\rank( \widehat{ \Sigma}_{n}^{ \lambda}) = r$ with probability at least $1 - d^{-1}$ if both $\log(d) = o( \Delta_{n}^{-1})$ and $\Delta_{n}^{-1}= o(d^{2+2 \gamma} \log(d))$, and $n$ is large. Or to put it differently, $\rank( \widehat{ \Sigma}_{n}^{ \lambda})$ is, with a probability tending to one, exactly the number of underlying factors in the model.
	
	We remark that the restrictions imposed above are too weak to ensure that the Frobenius estimation error $\lVert \widehat{ \Sigma}_{n}^{ \lambda} - \Sigma \rVert_{2}$ tends to zero, meaning that our estimator can be used to detect the underlying factor structure even in very high-dimensional settings. To ensure the estimation error is small, we need $d^{1+2 \gamma} \log(d) = o( \Delta_{n}^{-1})$ rather than $\log(d) = o( \Delta_{n}^{-1})$.
\end{example}

\subsection{Selection of tuning parameter} \label{section:tuning-parameter}

In view of Theorem \ref{lowerBound} and Corollary \ref{rankConcentration}, it follows that $\widehat{ \Sigma}_{n}^{ \lambda}$ can be of low rank and accurate, given optimal tuning of $\lambda$. However, as evident from \eqref{rankRegularization}, $\lambda$ depends on the latent spot variance process $(c_t)_{t \in [0,1]}$ through $\nu_{c,2}$ and $\nu_{c, \infty}$ as well as the unknown absolute constant $\gamma$, whose value can be important in finite samples.

We remark that $\widehat{ \Sigma}_{n} - \Sigma$ provides an estimate of the null matrix, but it is perturbed by randomness in the data. Hence, a good choice of shrinkage parameter exactly shrinks the eigenvalues of $\widehat{ \Sigma}_{n} - \Sigma$ to zero (in view of problem \eqref{equation:lasso-estimation} with $\widehat{ \Sigma}_{n}$ replaced by $\widehat{ \Sigma}_{n} - \Sigma$). By Proposition \ref{proposition:soft-thresholding}, this means $\lambda = 2 \lVert \widehat{ \Sigma}_{n} - \Sigma \rVert_{ \infty}$.

The above is unobservable. To circumvent this problem and facilitate the calculation of our estimator in Section \ref{section:simulation} and \ref{section:empirical}, we adopt a data-driven shrinkage selection by exploiting the subsampling technique of \citet*{christensen-podolskij-thamrongrat-veliyev:17a}, see also \citet*{politis-romano-wolf:99a} and \cite*{kalnina:11a}. To explain the procedure in short, suppose for notational convenience that $n = \Delta_{n}^{-1}$. We select an integer $L$---the number of subsamples---that divides $n$ and assign log-returns successively to each subsample. Hence, the $l$th subsample consists of the increments $\bigl( \Delta_{(k - 1)L + l}^{n} Y \bigr)_{k = 1, \ldots, n/L}$ for $l = 1, \dots, L$. We denote the associated subsampled RV by
\begin{equation}
\widehat{\Sigma}_{n,l} = \frac{1}{n/L} \sum_{k=1}^{n/L} ( \sqrt{n} \Delta_{(k-1)L+l}^{n} Y) ( \sqrt{n} \Delta_{(k-1)L+l}^{n} Y)^{ \top},
\end{equation}
Note that the random matrices in the sequence $(\widehat{ \Sigma}_{n,l})_{l = 1}^{L}$ are asymptotically conditionally independent, as $n \rightarrow \infty$. Moreover, it follows from \citet*{christensen-podolskij-thamrongrat-veliyev:17a} that as $n \rightarrow \infty$, $L \rightarrow \infty$ and $n/L \rightarrow \infty$, the half-vectorization of $\sqrt{n/L}( \widehat{\Sigma}_{n,l} - \widehat{\Sigma}_{n})$ and $\sqrt{n}(\widehat{\Sigma}_{n} - \Sigma)$ converge in law to a common mixed normal distribution, derived in \citet*[][Theorem 1]{barndorff-nielsen-shephard:04a}.

Hence, in each subsample $\lambda = 2 \lVert \widehat{\Sigma}_{n} - \Sigma \rVert_{ \infty}$ can be approximated by
\begin{equation}
\lambda_{l} = \frac{2}{ \sqrt{L}} \lVert \widehat{ \Sigma}_{n,l} - \widehat{ \Sigma}_{n} \rVert_{ \infty}
\end{equation}
As an estimator of $\lambda$, we therefore take the sample average:
\begin{equation}
\lambda^{ \ast} = \frac{1}{L} \sum_{l=1}^{L} \lambda_{l}.
\end{equation}

\section{Estimation of the local variance} \label{section:local-volatility}

As noted in Section \ref{section:introduction}, the rank of the local variance $c_{t}$ is often much smaller than the rank of $\Sigma$. In fact, although $\Sigma$ may be well-approximated by a matrix of low rank, we should expect that $\rank (\Sigma) = d$. On the other hand, there are many situations where it is reasonable to expect that $\rank(c_{t})$ is small, in which case it provides valuable insight about the complexity of the underlying model. This motivates developing a theory for the spot version of our penalized estimator with particular attention on its ability to identify $\rank (c_t)$.

To estimate $c_{t}$, we follow \citet*{jacod-protter:12a} and apply a localized realized variance, which is defined over a block of $\lfloor h_{n} / \Delta_{n} \rfloor \geq 2$ log-price increments with $h_{n} \in (0,1)$. The RV computed over the time window $[t, t+h_{n}]$, for $t \in (0,1-h_{n}]$, is then defined as follows:
\begin{equation}
\widehat{ \Sigma}_{n}(t;t+h_n) = \sum_{k = \lfloor t / \Delta_{n} \rfloor+1}^{ \lfloor(t+h_{n}) / \Delta_{n} \rfloor}
( \Delta_{k}^{n} Y)( \Delta_{k}^{n} Y)^{ \top}.
\end{equation}
The corresponding penalized version $\widehat{ \Sigma}_{n}^{ \lambda} (t;t+h_{n})$ is computed as
\begin{equation}
\widehat{ \Sigma}_{n}^{ \lambda} (t;t+h_{n}) = \argmin_{A \in \mathbb{R}^{d \times d}} \bigl( \lVert \widehat{ \Sigma}_n(t;t+h_{n}) - A \rVert_{2}^{2} + \lambda \lVert A \rVert_{1} \bigr),
\end{equation}
or by using Proposition \ref{proposition:soft-thresholding} with $\widehat{ \Sigma}_{n}$ replaced by $\widehat{ \Sigma}_{n}(t;t+h_{n})$. Then, the penalized estimator $\hat{c}_{n}^{ \lambda}(t)$ of $c_{t}$ is given by
\begin{equation} \label{penVolatility}
\hat{c}_{n}^{ \lambda}(t) = \frac{ \widehat{ \Sigma}_{n}^{ \lambda} (t;t+h_{n})}{h_{n}}.
\end{equation}
Also, we write
\begin{equation}
\Sigma(t_{1};t_{2}) = \int_{t_{1}}^{t_{2}} c_{s} \mathrm{d}s,
\end{equation}
for the QV over $(t_{1},t_{2}]$ for arbitrary time points $t_{1},t_{2} \in [0,1]$ with $t_{1} < t_{2}$.

Recall that for a convex and increasing function $\psi \colon [0, \infty) \rightarrow [0, \infty)$ with $\psi(0) = 0$, the Orlicz norm of a random variable $Z$ with values in the Hilbert space $( \mathbb{R}^{d \times d}, \lVert \: \cdot \: \rVert_{2})$ is defined as:
\begin{equation}
\lVert Z \rVert_{ \psi} \equiv \inf \bigl\{c > 0 : \mathbb{E}[ \psi( \lVert Z \rVert_{2}/c)] \leq 1 \bigr\}.
\end{equation}
This setting includes, in particular, the $L^{p}( \mathbb{P})$ norms (with $\psi(s) = s^{p}$), but also the $\psi_{1}$ and $\psi_{2}$ norms for sub-exponential ($\psi(s) = e^{s}-1$) and sub-Gaussian ($\psi(s) = e^{s^{2}}-1$) random variables. For convenience, we further impose the mild restriction that the range of $\psi$ is $[0, \infty)$, so it admits an inverse $\psi^{-1}$ on $[0, \infty)$.

\begin{theorem} \label{localVol}
	Suppose that Assumption \ref{assumption:A} holds with $\nu_{ \mu} \leq \nu_{c, \infty} \leq \nu_{c,2}$. Furthermore, let $t \in \big[ \frac{h_n}{2},1- \frac{3h_n}{2} \big]$ and assume $\sup_{0 \leq u<s \leq 1} \lVert c_{s} - c_{u} \rVert_{ \psi}/ \sqrt{s-u} \leq \nu_{c, \psi}$ and $h_{n} \geq 2 \Delta_{n} \frac{ \nu_{c,2}}{ \nu_{c, \infty}} \log(d)$. Take the regularization parameter $\lambda$ as
	\begin{equation} \label{localLambda}
	\lambda = \gamma \sqrt{h_{n} \Delta_{n} \nu_{c,2} \nu_{c, \infty} \log(d)}
	\end{equation}
	for a sufficiently large absolute constant $\gamma$. Then,
	\begin{equation} \label{localInequality}
	\lVert \hat{c}_{n}^{ \lambda}(t) - c_{t} \rVert_{2}^{2} \leq \kappa \gamma^{2} \bigg( \frac{ \nu_{c,2} \nu_{c, \infty} \Delta_{n} \rank \big( \Sigma(t- \frac{h_n}{2};t+ \frac{3h_n}{2}) \big)}{h_{n}} + h_{n} \nu_{c, \psi}^{2} \bigg) \log(d),
	\end{equation}
	with probability at least $1-d^{-1} - \psi( \sqrt{ \log(d)})^{-1}$ for some absolute constant $\kappa$.
\end{theorem}
The bound on the estimation error of $\hat{c}_{n}^{ \lambda}(t)$ builds on Corollary \ref{consequenceOfPenConcentration}, but Theorem \ref{localVol} further enforces a smoothness condition on $(c_{t})_{t \in [0,1]}$. It entails a trade-off between the smoothness of spot variance and the concentration probability associated with \eqref{localInequality}. For example, if $(c_{t})_{t \in[0,1]}$ is $1/2$-H\"{o}lder continuous in $L^{2}( \mathbb{P})$, which corresponds to setting $\psi(s) = s^{2}$, then we end up with a concentration probability converging to one at rate $\log(d)^{-1}$, which is slower than for the PRV. On the other hand, if $(c_{t})_{t \in[0,1]}$ is $1/2$-H\"{o}lder continuous in the sub-Gaussian norm $\psi(s) = e^{s^{2}}-1$, the concentration probability converges to one at the rate $d^{-1}$, which is equivalent to the PRV.

Note that with $d$ fixed, \eqref{localInequality} reveals how to select the length of the estimation window $h_{n}$ optimally. The upper bound depends on $\Delta_{n}/h_{n}$ and $h_{n}$, and for these terms to converge to zero equally fast, we should take $h_{n} \asymp \sqrt{ \Delta_{n}}$. This is consistent with the literature on spot variance estimation; see, e.g., \citet*{jacod-protter:12a}.

The next result, which relies on Corollary \ref{rankConcentration}, shows that the rank and performance of $\hat{c}_{n}^{ \lambda}(t)$ can be controlled simultaneously.

\begin{theorem}\label{localVol:rank}
	Suppose that Assumption \ref{assumption:A} holds with $\nu_{ \mu} \leq \nu_{c, \infty} \leq \nu_{c,2}$ and fix $\delta \in (0,1/4)$. Furthermore, let $t \in [ \frac{h_n}{2},1 - \frac{3h_n}{2}]$ and assume that $\sup_{0 \leq u < s \leq 1} \lVert c_{s} - c_{u} \rVert_{ \psi} / \sqrt{s-u} \leq \nu_{c, \psi}$. Suppose further that $h_{n} \geq 2 \Delta_{n} \frac{ \nu_{c,2}}{ \nu_{c, \infty}} \log(d)$ and consider a regularization parameter $\lambda$ such that
	\begin{equation}\label{localLambda2}
	\lambda = \gamma \sqrt{h_{n} \Delta_{n} \nu_{c,2} \nu_{c, \infty} \log(d)},
	\end{equation}
	for a sufficiently large constant $\gamma$ depending only on $\delta$. Then, with probability at least $1-d^{-1}- \psi ( \sqrt{ \log(d)})^{-1}$, it holds that
	\begin{equation}\label{rankInequality}
	\lVert \hat{c}^{ \lambda}_{n}(t) - c_{t} \rVert_{2}^{2} \leq \kappa \gamma^{2} \biggl( \frac{ \nu_{c,2} \nu_{c, \infty} \Delta_{n} \rank \bigl( \Sigma(t- \frac{h_{n}}{2};t+ \frac{3h_{n}}{2}) \bigr)}{h_{n}} + h_{n} \nu_{c, \psi}^{2} \biggr) \log(d),
	\end{equation}
	and
	\begin{equation} \label{rankBounds_theorem}
	\rank (c_{t}; \overline{ \varepsilon}) \leq \rank( \hat{c}_{n}^{ \lambda}(t)) \leq \rank(c_{t}; \underline{ \varepsilon}),
	\end{equation}
	where
	\begin{align}
	\begin{split}
	\underline{ \varepsilon} &\equiv \delta \gamma \max \biggl \{ \sqrt{ \frac{ \nu_{c,2} \nu_{c, \infty} \Delta_{n}}{h_{n}}} - \nu_{c, \psi} \sqrt{h_{n}},0 \biggr \} \sqrt{ \log(d)}, \\[0.10cm]
	\overline{ \varepsilon} &\equiv \gamma \biggl( \sqrt{ \frac{ \nu_{c,2} \nu_{c, \infty} \Delta_{n}}{h_{n}}} + \nu_{c, \psi} \sqrt{h_{n}} \biggr) \sqrt{ \log(d)},
	\end{split}
	\end{align}
	and $\kappa$ is an absolute constant. If, in addition, $\underline{ \varepsilon} > 0$ and $\overline{ \varepsilon} \leq s_{ \rank(c_{t})} (c_{t})$, then both \eqref{rankInequality} and 	$\rank (\hat{c}_{n}^{ \lambda}(t)) = \rank(c_{t})$ hold with probability at least $1-d^{-1}- \psi ( \sqrt{ \log(d)})^{-1}$.	
\end{theorem}

\noindent Consider the setting of Theorem~\ref{localVol:rank}. For the upper bound in \eqref{rankBounds_theorem} to be useful, we must have that $\underline{ \varepsilon} > 0$ or, equivalently,
\begin{equation} \label{windowRestriction}
\nu_{c, \psi}^{2} h_{n}^{2} < \nu_{c,2} \nu_{c, \infty} \Delta_{n}.
\end{equation}
Suppose that $\nu_{c,2}$, $\nu_{c,\infty}$, and $\nu_{c,\psi}$ do not depend on $n$. The inequality \eqref{windowRestriction} can then be achieved in large samples by choosing $h_{n} = o( \sqrt{ \Delta_{n}})$. However, as pointed out above one should take $h_{n} \asymp \sqrt{ \Delta_{n}}$ in order to achieve an optimal rate for the estimation error $\lVert \hat{c}_{n}^{ \lambda}(t) - c_{t} \rVert_{2}^{2}$. In that case, \eqref{windowRestriction} translates directly into an upper bound on $\nu_{c, \psi}$, which concerns the smoothness of $(c_t)_{t \in [0,1]}$.

When can the term $\rank( \Sigma (t- \frac{h_{n}}{2}; t+ \frac{3h_{n}}{2}))$ in the upper bound of \eqref{rankInequality} be replaced by $\rank(c_t)$? This question appears difficult to answer
in general, but it is not too difficult to show that the replacement is valid if the volatility process is locally constant.

\section{Simulation study} \label{section:simulation}

In this section, we do a Monte Carlo analysis to inspect the finite sample properties of the PRV, $\widehat{\Sigma}_n^\lambda$. The aim here is to demonstrate the ability of our estimator to detect the presence of near deficient rank in a large-dimensional setting. In view of Theorem \ref{localVol:rank}, this can be achieved by a proper selection of $\lambda$.
We simulate a $d = 30$-dimensional log-price process $Y_{t}$. The size of the vector corresponds to the number of assets in our empirical investigation. We assume $Y_{t}$ follows a slightly modified version of the $r = 3$-factor model proposed in \citet*{ait-sahalia-xiu:19b}:
\begin{equation} \label{equation:model}
\mathrm{d} Y_{t} = \beta \mathrm{d} F_{t} + \mathrm{d} Z_{t}, \quad t \in [0,1],
\end{equation}
where $F_{t} \in \mathbb{R}^{r}$ is a vector of systematic risk factors, which has the dynamic:
\begin{equation}
\mathrm{d} F_{t} = \alpha \mathrm{d} t + \sigma_{t} L \mathrm{d} W_{t},
\end{equation}
and $W_{t} \in \mathbb{R}^{r}$ is a standard Brownian motion. The drift term $\alpha \in \mathbb{R}^{r}$ is constant $\alpha = (0.05,0.03,0.02)^{ \top}$. The random volatility matrix $\sigma_{t} = \diag ( \sigma_{1,t}, \sigma_{2,t}, \sigma_{3,t}) \in \mathbb{R}^{r \times r}$, where $\diag ( \: \cdot\:)$ is a diagonal matrix with coordinates $\:\cdot\:$, is simulated as a \citet*{heston:93a} square-root process:
\begin{equation} \label{equation:factor-volatility}
\mathrm{d} \sigma_{j,t}^{2} = \kappa_{j} ( \theta_{j} - \sigma_{j,t}^{2}) \mathrm{d}t + \eta_{j} \sigma_{j,t} \Big( \rho_{j} \mathrm{d} W_{j,t} + \sqrt{1 - \rho_{j}^{2}} \mathrm{d} \widetilde{W}_{j,t} \Big),
\end{equation}
for $j = 1, \dots, r$, where $\widetilde{W}_{t} \in \mathbb{R}^{r}$ is an $r$-dimensional standard Brownian motion independent of $W_{t}$. As in \citet*{ait-sahalia-xiu:19b}, the parameters are $\kappa = (3, 4, 5)^{ \top}$, $\theta = (0.05, 0.04, 0.03)^{ \top}$, $\eta = (0.3, 0.4, 0.3)^{ \top}$, and $\rho = (-0.60, -0.40, -0.25)^{ \top}$.
The factor correlation is captured by the Cholesky component $L$:
\begin{equation}
\rho = L L^{ \top} = \begin{bmatrix}
1.00 & 0.05 & 0.10 \\
0.05 & 1.00 & 0.15 \\
0.10 & 0.15 & 1.00
\end{bmatrix}.
\end{equation}
The idiosyncratic component $Z_{t} \in \mathbb{R}^{d}$ is given by
\begin{equation}
\mathrm{d} Z_{t} = \sqrt{g_t} \mathrm{d} B_{t},
\end{equation}
where $g_{t} = \diag ( \gamma_{1,t}^2, \dots, \gamma_{d,t}^2)$ with
\begin{equation}
\mathrm{d} \gamma_{j,t}^{2} = \kappa_{Z} ( \theta_{Z} - \gamma_{j,t}^{2}) \mathrm{d}t + \eta_{Z} \gamma_{j,t} \mathrm{d} \widetilde{B}_{j,t},
\end{equation}
for $j = 1, \dots, d$, and $B_{t} \in \mathbb{R}^{d}$ and $\widetilde{B}_{t} \in \mathbb{R}^{d}$ are independent $d$-dimensional standard Brownian motions. Thus, $g_{t}$ is the instantaneous variance of the unsystematic component. We set $\kappa_{Z} = 4$, $\theta_{Z} = 0.25$, and $\eta_{Z} = 0.06$. Hence, the idiosyncratic error varies independent in the cross-section of assets, but the parameters controlling the dynamics are common.

\begin{figure}[t!]
\begin{center}
\caption{Relative importance of eigenvalues in simulated model. \label{figure:screeSim}}
\begin{tabular}{cc}
\small{Panel A: Across replica.} & \small{Panel B: On average.} \\
\includegraphics[height=8cm,width=0.48\textwidth]{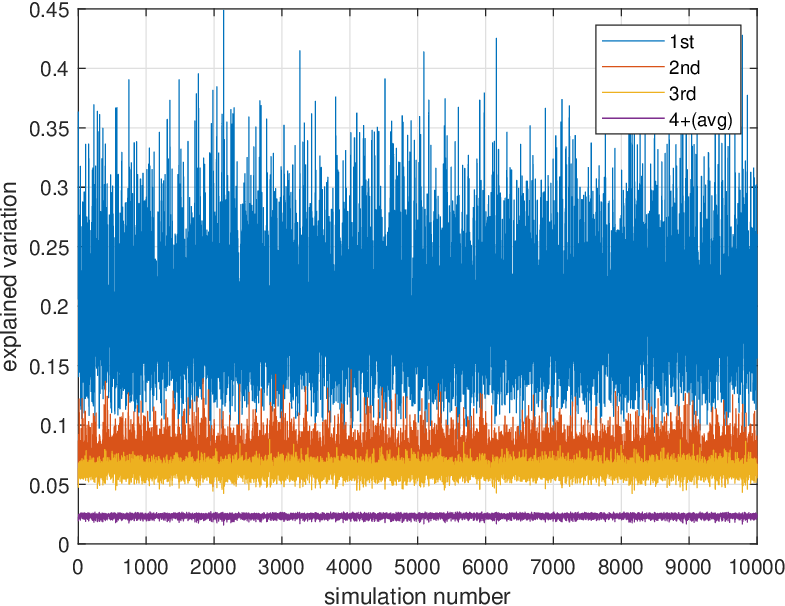} &
\includegraphics[height=8cm,width=0.48\textwidth]{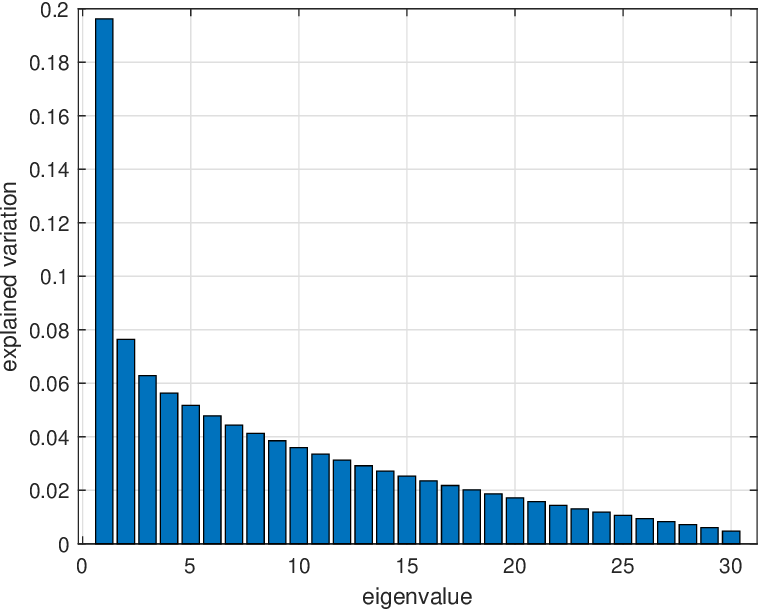}
\end{tabular}
\begin{scriptsize}
\parbox{\textwidth}{\emph{Note.} In Panel A, we plot for daily RV the relative size of the three largest eigenvalues, and also the average of the remaining 27 eigenvalues. In Panel B, we plot a histogram of the sample average of the relative size of each of the eigenvalues across simulations.}
\end{scriptsize}
\end{center}
\end{figure}

The above setup implies
\begin{equation} \label{covStructure}
c_{t} = \beta \sigma_{t} \rho \sigma_{t} \beta^{ \top} + g_{t},
\end{equation}
so the spot covariance matrix has full rank at every time point $t$. However, it has only $r = 3$ large eigenvalues associated with the systematic factors, while the remaining $d-r$ associated with the idiosyncratic variance are relatively small. Note that in contrast to \citet*{ait-sahalia-xiu:19b}, we allow for time-varying idiosyncratic variance.

We set $\Delta_{n} = 1/n$ with $n = 78$. This corresponds to 5-minute sampling frequency within a 6.5 hour window. Hence, $d$ is relatively large compared to $n$.
We construct 10,000 independent replications of this model. At the beginning of each simulation, we draw the associated factor loadings $\beta \in \mathbb{R}^{d \times r}$ at random such that the first column (interpreted as the loading on the market factor) are from a uniform distribution on $(0.25,1.75)$, i.e. $\beta_{i1} \sim U(0.25,1.75)$, $i = 1, \dots, d$. The range of the support is consistent with the realized beta values reported in Table \ref{table:penalized-rv} in Section \ref{section:empirical}. The remaining columns are generated as $\beta_{ij} \sim N(0, 0.5^{2})$, $i = 1, \dots, d$ and $j = 2$ and $3$.

In Panel A of Figure \ref{figure:screeSim}, we plot the relative size of the three largest eigenvalues, extracted from the corresponding RV, together with the average of the remaining 27 eigenvalues. In Panel B, we include a histogram of the relative size of the eigenvalues across simulations. We observe that it is generally challenging to distinguish important factors from idiosyncratic variation, since the relative size of the eigenvalues decays smoothly with the exception of the largest (and perhaps the second largest) eigenvalue.

In each simulation, we employ the subsampling procedure described in Section \ref{section:tuning-parameter} with $L = 6$ to choose $\lambda^{*}$.\footnote{We also employed $L = 13$ subsamples, but there were no discernible change in the results.}

\begin{figure}[t!]
\begin{center}
\caption{Properties of the PRV. \label{figure:simulation-rank}}
\begin{tabular}{cc}
\small{Panel A: Shrinkage parameter.} & \small{Panel B: Rank.} \\
\includegraphics[height=8cm,width=0.48\textwidth]{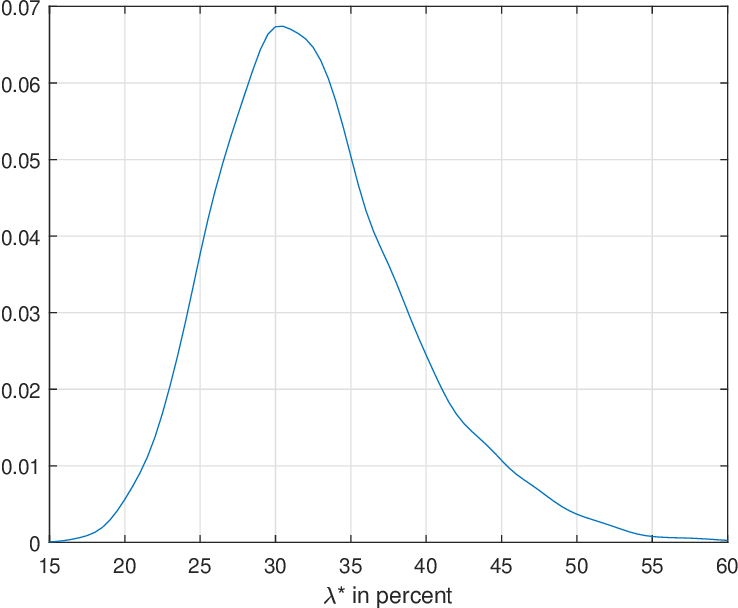} &
\includegraphics[height=8cm,width=0.48\textwidth]{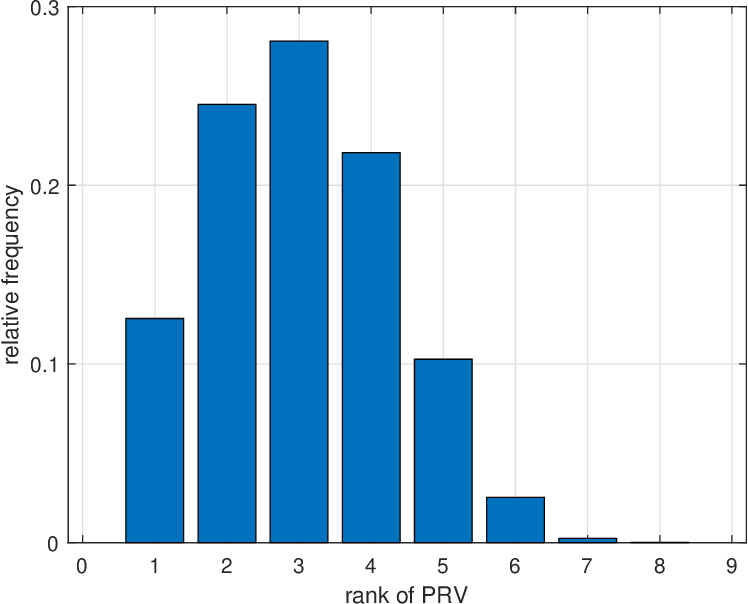}
\end{tabular}
\begin{scriptsize}
\parbox{\textwidth}{\emph{Note.} In Panel A, we report a kernel density estimate of the shrinkage parameter $\lambda$, expressed in percent of the maximal eigenvalue of the RV matrix. In Panel B, we show the relative frequency histogram of the associated rank of the PRV.}
\end{scriptsize}
\end{center}
\end{figure}

In Panel A of Figure \ref{figure:simulation-rank}, we report the distribution of the $\lambda^{ \ast}$ parameter across simulations, which is here expressed in percent of the maximal eigenvalue of RV. Overall, the penalization is relatively stable with a tendency to shrink around one-fourth to one-third of the largest eigenvalue most of the times.

In Panel B, we plot a histogram of the relative frequency of the rank of the PRV. In general, the PRV does a good job at identifying the number of driving factors, given the challenging environment. Although there are variations in the rank across simulations, caused by the various sources of randomness incorporated in our model, the picture is close to the truth. This means we tend to identify about one-to-three driving factors. The average rank estimate is 1.79, so as expected there is a slight tendency to overshrink leading to a small downward bias. Based on these findings, we can confidently apply the PRV in the empirical application.

\section{Empirical application} \label{section:empirical}

In this section, we apply the PRV to an empirical dataset. The sample period is from January 3, 2007 through May 29, 2020 (3,375 trading days in total) and includes both the financial crisis around 2007 -- 2009 and partly the recent events related to Covid-19.

\begin{table}[H]
\begin{footnotesize}\setlength{ \tabcolsep}{0.70cm}
\begin{center}
\caption{Descriptive statistics of high-frequency data.
\label{table:penalized-rv}}
\vspace*{-0.25cm}
\begin{tabular}{lrcccccccccc}
\hline \hline
&&&& \multicolumn{2}{c}{realized beta} \\
\cline{5-6}
code & \multicolumn{1}{c}{$N$} & $\sigma_{ \text{RV}}$ & $\sigma_{ \epsilon}$ & median & [$q_{0.05}, q_{0.95}$] \\
\hline 
AAPL & 171,438 & 0.227 & 0.165 & 0.990 & [0.539, 1.784] \\
AXP & 42,379 & 0.279 & 0.207 & 1.016 & [0.532, 1.730] \\
BA & 41,579 & 0.249 & 0.192 & 0.918 & [0.427, 1.489] \\
CAT & 41,943 & 0.256 & 0.187 & 1.154 & [0.570, 1.747] \\
CSCO & 111,876 & 0.223 & 0.163 & 0.997 & [0.498, 1.470] \\
CVX & 54,037 & 0.223 & 0.163 & 0.887 & [0.370, 1.407] \\
DIS & 56,563 & 0.213 & 0.157 & 0.863 & [0.411, 1.303] \\
DOW & 35,108 & 0.240 & 0.179 & 1.002 & [0.467, 1.567] \\
GS & 42,098 & 0.287 & 0.216 & 1.091 & [0.572, 1.812] \\
HD & 50,777 & 0.222 & 0.163 & 0.878 & [0.422, 1.343] \\
IBM & 35,955 & 0.182 & 0.131 & 0.769 & [0.402, 1.136] \\
INTC & 126,855 & 0.233 & 0.174 & 1.026 & [0.531, 1.586] \\
JNJ & 54,527 & 0.151 & 0.120 & 0.534 & [0.168, 0.982] \\
JPM & 127,497 & 0.296 & 0.216 & 1.135 & [0.680, 1.911] \\
KO & 55,932 & 0.156 & 0.126 & 0.496 & [0.116, 0.891] \\
MCD & 36,777 & 0.168 & 0.134 & 0.523 & [0.137, 0.902] \\
MMM & 24,253 & 0.183 & 0.131 & 0.825 & [0.426, 1.201] \\
MRK & 60,637 & 0.197 & 0.161 & 0.695 & [0.241, 1.166] \\
MSFT & 162,004 & 0.211 & 0.154 & 0.957 & [0.501, 1.543] \\
NKE & 32,390 & 0.216 & 0.171 & 0.806 & [0.360, 1.260] \\
PFE & 97,351 & 0.196 & 0.159 & 0.726 & [0.281, 1.184] \\
PG & 53,787 & 0.157 & 0.128 & 0.490 & [0.117, 0.918] \\
RTX & 31,372 & 0.203 & 0.149 & 0.861 & [0.445, 1.248] \\
TRV & 18,855 & 0.219 & 0.175 & 0.688 & [0.220, 1.190] \\
UNH & 37,540 & 0.251 & 0.212 & 0.771 & [0.247, 1.307] \\
V & 44,085 & 0.228 & 0.188 & 0.835 & [0.322, 1.336] \\
VZ & 69,710 & 0.185 & 0.146 & 0.560 & [0.139, 1.033] \\
WBA & 37,302 & 0.218 & 0.181 & 0.786 & [0.261, 1.308] \\
WMT & 59,091 & 0.165 & 0.135 & 0.525 & [0.154, 0.906] \\
XOM & 90,346 & 0.203 & 0.146 & 0.809 & [0.332, 1.261] \\
\hline \hline
\end{tabular}
\smallskip
\begin{scriptsize}
\parbox{0.98\textwidth}{\emph{Note.} 
``code'' is the ticker symbol, 
``$N$'' is the number of transactions, 
``$\sigma_{ \text{RV}}$'' is the annualized square-root realized variance, 
``$\sigma_{ \epsilon}$'' is the standard deviation of the idiosyncratic component. 
The latter is computed as $\sigma_{ \epsilon}^{2} = \sigma_{ \text{RV}}^{2} - \beta^{2} \sigma_{\text{SPY}}^{2}$, 
where $\beta$ is the realized beta between the asset and SPY (acting as market portfolio). 
}
\end{scriptsize}
\end{center}
\end{footnotesize}\end{table}

We look at high-frequency data from the 30 members of the Dow Jones Industrial Average (DJIA) index as of April 6, 2020.\footnote{On April 6, 2020, Raytheon Technologies (RTX) replaced United Technologies (UTX) in the DJIA index following a merger of United Technologies and Raytheon Company. Moreover, on April 2, 2019, Dow replaced DowDuPont following a spin-off from the parent company, which itself was a fusion between Dow Chemical Company and DuPont in 2017. We employ high-frequency data for the preceding member prior to each index update.} The ticker codes of the various firms are listed in Table \ref{table:penalized-rv}, along with selected descriptive statistics. As readily seen from Table \ref{table:penalized-rv}, most of these companies are very liquid. In order to compute a daily RV matrix estimate we collect the individual 5-minute log-return series spanning the 6.5 hours of trading on U.S. stock exchanges from 9:30am to 4:00pm EST.\footnote{We truncate 5-minute log-returns that exceed three local standard deviations, as gauged by the daily bipower variation estimator, in order to get shelter from potential jumps.} Hence, our analysis is based on a relatively large cross-section (i.e., $d = 30$) compared to the sampling frequency (i.e., $n = 78$). The realized beta in Table \ref{table:penalized-rv} is calculated with SPY acting as market portfolio. The latter is an exchange-traded fund that tracks the S\&P500 and its evolution is representative of the overall performance of the U.S. stock market. The dispersion of realized beta is broadly consistent with the simulated values generated in Section \ref{section:simulation}.

\begin{figure}[t!]
\begin{center}
\caption{Proportion of variance explained by each eigenvalue. \label{figure:scree}}
\begin{tabular}{cc}
\small{Panel A: Over time.} & \small{Panel B: On average.} \\
\includegraphics[height=8cm,width=0.48\textwidth]{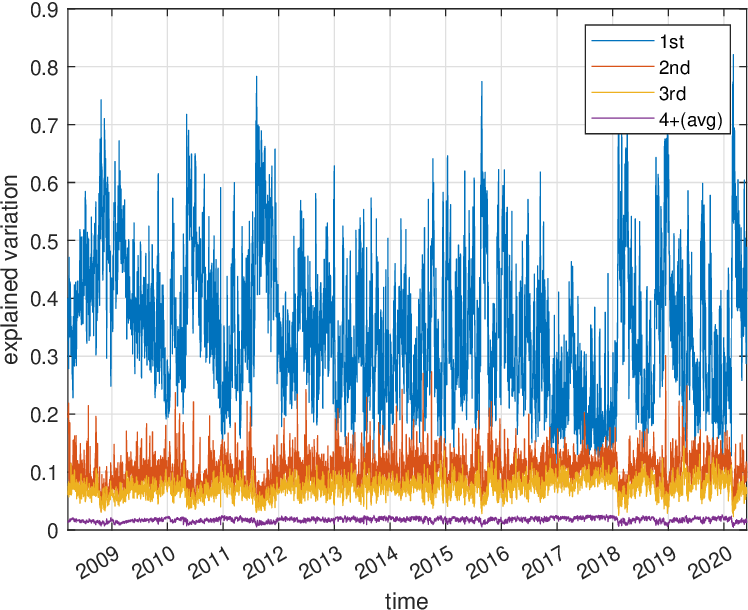} &
\includegraphics[height=8cm,width=0.48\textwidth]{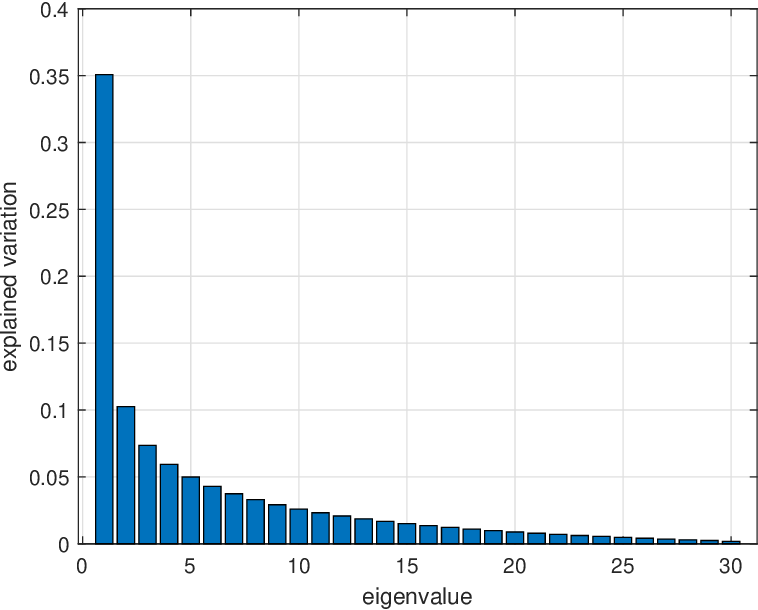}
\end{tabular}
\begin{scriptsize}
\parbox{\textwidth}{\emph{Note.} In Panel A, we plot for each day in the sample the relative size of the three largest eigenvalues of RV, as well as the average of the remaining 27. In Panel B, we plot a histogram of the time series average of the relative size of each eigenvalues.}
\end{scriptsize}
\end{center}
\end{figure}

In Figure \ref{figure:scree}, we depict the factor structure in the time series of RV. In Panel A, we compute on a daily basis the proportion of the total variance (i.e., trace) explained by each eigenvalue, whereas Panel B reports the sample average of this statistic. As consistent with \citet*{ait-sahalia-xiu:19b}, we observe a pronounced dynamic evolution in the contribution of each eigenvalue to RV with notably changes corresponding to times of severe distress in financial markets. The first factor captures the vast majority of aggregate return variation (about 35\% on average). It is followed by a few smaller---but still important---factors accountable for an incremental 25\% -- 30\% of the total variance, whereas the last 25 or so eigenvalues are relatively small.

\begin{figure}[t!]
\begin{center}
\caption{Rank of PRV. \label{figure:rank}}
\begin{tabular}{cc}
\includegraphics[height=8cm,width=0.48\textwidth]{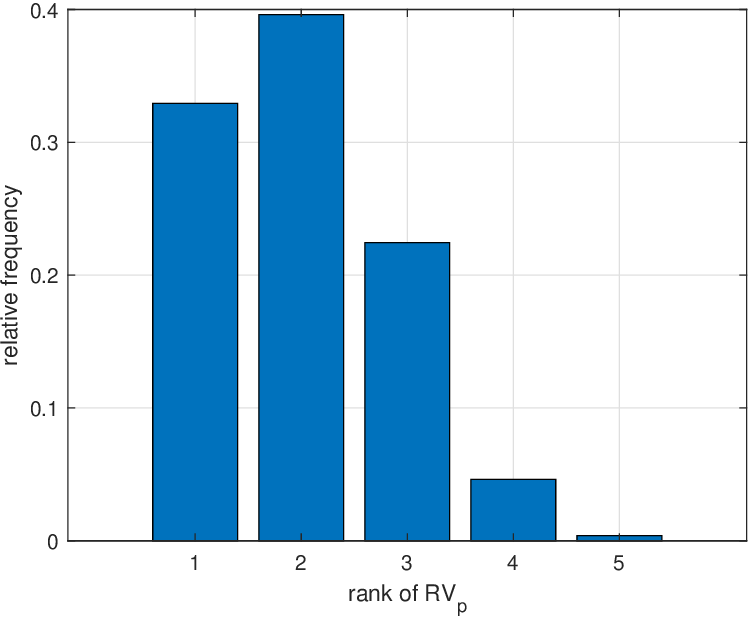} &
\includegraphics[height=8cm,width=0.48\textwidth]{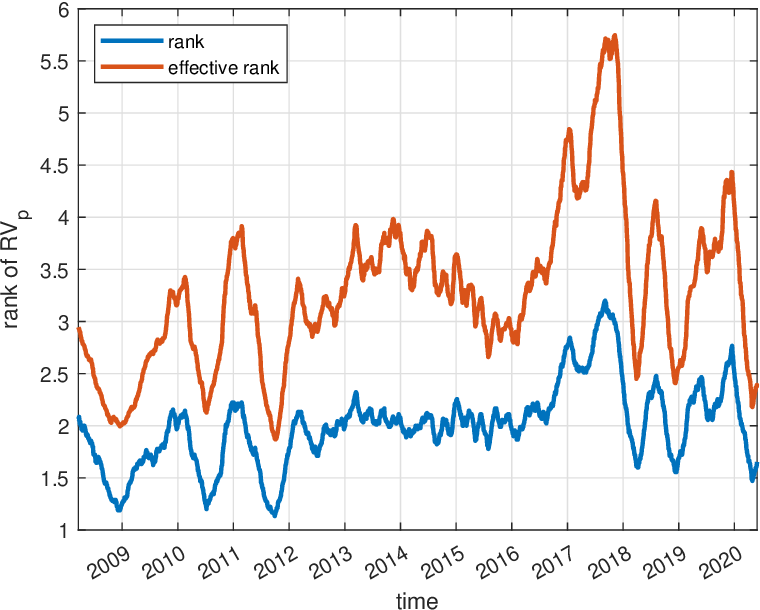}
\end{tabular}
\begin{scriptsize}
\parbox{\textwidth}{\emph{Note.} In Panel A, we report the relative frequency of the rank of the PRV. In Panel B, we compare the estimated rank to the effective rank, where both are smoothed over a three-month moving average window.}
\end{scriptsize}
\end{center}
\end{figure}

Next, we turn our attention to the PRV estimator. To select the shrinkage parameter $\lambda$, we follow the subsampling implementation from the simulation section. The relative frequency histogram of the rank of the PRV is reported in Panel A of Figure \ref{figure:rank}, whereas Panel B reports a three-month (90-day) moving average of the rank. The vast majority (around 95\%) of the daily rank estimates are between one and three, which is consistent with a standard Fama--French three-factor interpretation of the data. There are relatively few rank estimates at four, and it never exceeds five (with about a dozen of the latter). In Panel B, we observe the rank varies over time in an intuitive fashion, often dropping close to a single-factor representation during times of crisis, where correlations tend to rise. As a comparison, we compute the effective rank of the RV (cf. \eqref{effectiveRank}), which does not depend on a shrinkage parameter. There is a high association between the series, which is consistent with a soft-thresholding that eliminates smaller eigenvalues of the RV.

\section{Conclusion}

In this paper, we develop a novel and powerful penalized realized variance estimator, which is applicable to estimate the quadratic variation of high-dimensional continuous-time semimartingales under a low rank constraint. Our estimator relies on regularization and adapts the principle ideas of the LASSO from regression analysis to the field of high-frequency volatility estimation in a high-dimensional setting. We derive a non-asymptotic analysis of our estimator, including bounds on its estimation error and rank. The estimator is found to be minimax optimal up to a logaritmic factor. We design a completely data-driven procedure for selecting the shrinkage parameter based on a subsampling approach. In a simulation study, the estimator is found to possess good properties. In our empirical application, we confirm a low rank environment that is consistent with a three-factor structure in the cross-section of log-returns from the large-cap segment of the U.S. stock market.

\appendix

\section{Appendix of proofs} \label{appendix:proofs}

This appendix presents the proofs of the results from the main text.

\begin{proof}[Proof of Proposition \ref{proposition:soft-thresholding}]
	Since the map $A \mapsto L_{n} (A) \equiv \lVert \widehat{\Sigma}_{n} - A \rVert_{2}^{2} + \lambda \lVert A \rVert_1$ is strictly convex, the minimizer $A$ of $L_{n}$ is uniquely determined by the property that $0 \in \mathbb{R}^{d \times d}$ belongs to the set
	\begin{equation} \label{uniqueProperty}
	\partial L_{n}(A) = \Bigg\{2(A- \widehat{ \Sigma}_n) + \lambda \bigg( \sum_{k=1}^{d} u_{k}(A) v_{k}(A)^{ \top} + P_{S_{1}(A)^{ \perp}} W P_{S_{2}(A)^{ \perp}} \bigg) : \lVert W \rVert_{ \infty} \leq 1 \Bigg\}.
	\end{equation}
	Here, $\partial L_{n}(A)$ denotes the subdifferential of $L_{n}$ at $A$, and we employ \cite{watson:92a} to get the explicit expression of this set in \eqref{uniqueProperty}. Moreover, $\{u_{1}(A), \dots, u_{d}(A) \}$ and $\{v_{1}(A), \dots, v_{d}(A) \}$ are the vectors introduced in the notation paragraph associated with the SVD of $A$, whereas $S_{1}(A)$ and $S_{2}(A)$ are the corresponding linear spans, and $P_{S}$ is the projection matrix onto the subspace $S$. With $\widehat{ \Sigma}_{n}^{ \lambda}$ given as in \eqref{RVthreshold}, we find that $\partial L_{n}( \widehat{ \Sigma}_{n}^{ \lambda})$ coincides with the set of matrices $V$ of the form:
	\begin{equation} \label{uniqueProperty2}
	V = - 2 \sum_{k \colon s_{k} \leq \lambda/2}s_{k} u_{k}u_{k}^{ \top} + \lambda P_{U} W P_{U},
	\end{equation}
	where $W \in \mathbb{R}^{d \times d}$ is such that $\lVert W \rVert_{ \infty} \leq 1$, and $U$ is the linear span of $\{u_{k}\, :\, s_{k} \leq \lambda/2 \}$. With $W = \frac{2}{ \lambda} \sum_{k \colon s_{k} \leq \lambda/2}s_{k} u_{k} u_{k}^{ \top}$, it follows that $\lVert W \rVert_{ \infty} \leq 1$ by submultiplicativity of the norm and $P_{U} WP_{U} = W$. Hence, \eqref{uniqueProperty2} shows that $0 \in \partial L_{n} ( \widehat{ \Sigma}_{n}^{ \lambda})$, so $\widehat{ \Sigma}_{n}^{ \lambda}$ is the unique minimizer of $L_{n}$. \qed
\end{proof}

\begin{proof}[Proof of Proposition \ref{boundOn}]
	The structure of the proof is based similar to the one of Theorem~1 in \cite{koltchinskii-lounici-tsybakov:11a}, but is modified to fit our setting. Starting from the definition of $\widehat{ \Sigma}_{n}^{ \lambda}$, and since $L_{n}(A) = \lVert \widehat{ \Sigma}_{n} - A \rVert_{2}^{2} + \lambda \lVert A \rVert_1$, it follows that
	\begin{equation} \label{firstStep}
	L_{n}( \widehat{ \Sigma}_{n}^{ \lambda}) \leq L_{n}(A)
	\end{equation}
	for any $A \in \mathbb{R}^{d \times d}$. Note also that
	\begin{equation*}
	\lVert \widehat{ \Sigma}_{n} - A \rVert_{2}^{2} - \lVert A - \Sigma \rVert_{2}^{2} - \lVert \widehat{ \Sigma}_{n} - \widehat{ \Sigma}_{n}^{ \lambda} \rVert_{2}^{2} + \lVert \widehat{ \Sigma}_{n}^{ \lambda} - \Sigma \rVert_{2}^{2} = 2 \langle \widehat{ \Sigma}_{n}^{ \lambda} - A, \widehat{ \Sigma}_{n} - \Sigma \rangle.
	\end{equation*}
	By using the above identity and \eqref{firstStep},
	\begin{equation*}
	\lVert \widehat{ \Sigma}_{n}^{ \lambda} - \Sigma \rVert_{2}^{2} \leq \lVert A - \Sigma \rVert_{2}^{2} + 2 \langle \widehat{ \Sigma}_{n}^{ \lambda} - A, \widehat{ \Sigma}_{n} - \Sigma \rangle + \lambda \big( \lVert A \rVert_{1} - \lVert \widehat{ \Sigma}_{n}^{ \lambda} \rVert_{1} \big).
	\end{equation*}
	The first term in the minimum follows by observing that
	\begin{align*}
	2 \langle \widehat{ \Sigma}_{n}^{ \lambda} - A, \widehat{ \Sigma}_{n}- \Sigma \rangle + \lambda ( \lVert A \rVert_{1} - \lVert \widehat{ \Sigma}_{n}^{ \lambda} \rVert_{1}) &\leq 2 \lVert \widehat{ \Sigma}_{n}^{ \lambda} - A \rVert_{1} \lVert \widehat{ \Sigma}_{n}- \Sigma \rVert_{ \infty} + \lambda ( \lVert A \rVert_{1} - \lVert \widehat{ \Sigma}_{n}^{ \lambda} \rVert_{1}) \\[0.10cm] &\leq 2 \lambda \lVert A \rVert_{1},
	\end{align*}
	where we employ the trace duality $\langle A_{1}, A_{2} \rangle \leq \lVert A_{1} \rVert_{1} \lVert A_{2} \rVert_{ \infty}$, the triangle inequality, and the assumption that $2\lVert \widehat{ \Sigma}_{n}-\Sigma \rVert_\infty\leq \lambda$. To show the second part of the proposition, we observe that if $B$ is a subgradient of $L_{n}$ at $\widehat{ \Sigma}_{n}^{ \lambda}$, then by definition
	\begin{equation*}
	L_{n}(A) - L_{n}( \widehat{ \Sigma}_{n}^{ \lambda}) \geq \langle B,A - \widehat{ \Sigma}_{n}^{ \lambda} \rangle
	\end{equation*}
	or, equivalently,
	\begin{equation*}
	\lVert A \rVert_{2}^{2} + \lambda \lVert A \rVert_{1} - \lVert \widehat{ \Sigma}_{n}^{ \lambda} \rVert_{2}^{2} - \lambda \lVert \widehat{ \Sigma}_{n}^{ \lambda} \rVert_{1} \geq \langle B + 2 \widehat{ \Sigma}_{n}, A - \widehat{ \Sigma}_{n}^{ \lambda} \rangle
	\end{equation*}
	for all matrices $A$. This shows that $B + 2 \widehat{ \Sigma}_{n}$ is a subgradient of the function $A \mapsto \lVert A \rVert_{2}^{2} + \lambda \lVert A \rVert_{1}$ at $\widehat{ \Sigma}_{n}^{ \lambda}$ and, thus, $B = 2 \widehat{ \Sigma}_{n}^{ \lambda} - 2 \widehat{ \Sigma}_{n} + \lambda \widehat{V}$ for an appropriate $\widehat{V} \in \partial \lVert \widehat{ \Sigma}_{n}^{ \lambda} \rVert_{1}$. Moreover, because $\widehat{ \Sigma}_{n}^{ \lambda}$ is a minimizer of $L_{n}$, there must exist $B \in \partial L_{n}( \widehat{ \Sigma}_{n}^{ \lambda})$ such that $\langle B, \widehat{ \Sigma}_{n}^{ \lambda} -A \rangle \leq 0$ for all $A \in \mathbb{R}^{d \times d}$ (see, e.g., \cite[Section 2.4]{clarke:90a}). Combining these facts, we establish that
	\begin{equation}\label{keyIneq}
	2 \langle \widehat{ \Sigma}_{n}^{ \lambda}, \widehat{ \Sigma}_{n}^{ \lambda} - A \rangle + \lambda \langle \widehat{V}, \widehat{ \Sigma}_{n}^{ \lambda} - A \rangle \leq 2 \langle \widehat{ \Sigma}_{n}, \widehat{ \Sigma}_{n}^{ \lambda} - A \rangle
	\end{equation}
	for some $\widehat{V} \in \partial \lVert \widehat{ \Sigma}_{n}^{ \lambda} \rVert_{1}$ and any $A \in \mathbb{R}^{d \times d}$. Note the identity
	\begin{equation} \label{relation}
	2 \langle \widehat{ \Sigma}_{n}^{ \lambda} - \Sigma, \widehat{ \Sigma}_{n}^{ \lambda} - A \rangle = \lVert \widehat{ \Sigma}^{ \lambda}_{n} - \Sigma \rVert_{2}^{2} + \lVert \widehat{ \Sigma}^{ \lambda}_{n} - A \rVert_{2}^{2} - \lVert A - \Sigma \rVert_{2}^{2}.
	\end{equation}
	Now, fix $A \in \mathbb{R}^{d \times d}$ and consider a generic matrix $V \in \partial \lVert A \rVert_{1}$. Then, if we subtract $\langle 2 \Sigma + \lambda V, \widehat{ \Sigma}_{n}^{ \lambda} - A \rangle$ on both sides of \eqref{keyIneq}, exploit \eqref{relation}, and note that $\langle \widehat{V} - V, \widehat{ \Sigma}_{n}^{ \lambda} - A \rangle \geq 0$ by the monotonicity of subdifferentials (see \cite[Proposition 2.2.9]{clarke:90a}), we deduce that
	\begin{equation} \label{keyIneq2}
	\lVert \widehat{ \Sigma}_{n}^{ \lambda} - \Sigma \rVert_{2}^{2} + \lVert \widehat{ \Sigma}_{n}^{ \lambda} - A \rVert_{2}^{2}
	\leq \lVert A - \Sigma \rVert_{2}^{2} + \lambda \langle V,A - \widehat{ \Sigma}_{n}^{ \lambda} \rangle + 2 \langle \widehat{ \Sigma}_{n} - \Sigma, \widehat{ \Sigma}_{n}^{ \lambda} - A \rangle.
	\end{equation}
	We shall bound each of the last two terms on the right-hand side of \eqref{keyIneq2}, starting with $\langle V,A - \widehat{ \Sigma}_{n}^{ \lambda} \rangle$. According to \cite{watson:92a}, with $r = \rank(A)$ and $A = \sum_{k=1}^{r} s_{k} u_{k} v_{k}^{ \top}$ being the SVD of $A$, the subdifferential $\partial \lVert A\rVert_1$ has the characterization:
	\begin{equation*}
	\partial \lVert A \rVert_{1} = \bigg\{ \sum_{k=1}^{r} u_{k} v_{k}^{ \top} + P_{S_{1}^{ \perp}} W P_{S_{2}^{ \perp}} : \lVert W \rVert_{ \infty} \leq 1 \bigg\}.
	\end{equation*}
	Here, $S_{1}$ and $S_{2}$ denote the span of $\{u_{1}, \dots, u_{r} \}$ and $\{v_{1}, \dots, v_{r} \}$, respectively. By the polar decomposition, one finds that $W \in \mathbb{R}^{d \times d}$ with $\lVert W \rVert_{ \infty} = 1$ and $\langle W,P_{S_{1}^{ \perp}} \widehat{ \Sigma}_{n}^{ \lambda}P_{S_{2}^{ \perp}} \rangle = \lVert P_{S_{1}^{ \perp}} \widehat{ \Sigma}_{n}^{ \lambda} P_{S_{2}^{ \perp}} \rVert_{1}$. Fixing $V$ to be the subgradient associated with this choice of $W$,
	\begin{align} \label{term1}
	\begin{split}
	\langle V,A - \widehat{ \Sigma}_{n}^{ \lambda} \rangle &= \bigg\langle \sum_{k=1}^{r} u_{k} v_{k}^{ \top}, A - \widehat{ \Sigma}_{n}^{ \lambda} \bigg\rangle - \langle W, P_{S_{1}^{ \perp}} \widehat{ \Sigma}_{n}^{ \lambda} P_{S_{2}^{ \perp}} \rangle \\[0.10cm]
	&= \bigg\langle \sum_{k=1}^{r} u_{k} v_{k}^{ \top}, P_{S_{1}}(A - \widehat{ \Sigma}_{n}^{ \lambda}) P_{S_{2}} \bigg\rangle - \lVert P_{S_{1}^{ \perp}} \widehat{ \Sigma}_{n}^{ \lambda} P_{S_{2}^{ \perp}} \rVert_{1} \\[0.10cm]
	&\leq \lVert P_{S_{1}} ( \widehat{ \Sigma}_{n}^{ \lambda} - A) P_{S_{2}} \rVert_{1} - \lVert P_{S_{1}^{ \perp}} \widehat{ \Sigma}_{n}^{ \lambda} P_{S_{2}^{ \perp}} \rVert_{1} \\[0.10cm]
	&\leq \sqrt{r} \lVert \widehat{ \Sigma}_{n}^{ \lambda} - A \rVert_{2} - \lVert P_{S_{1}^{ \perp}} \widehat{ \Sigma}_{n}^{ \lambda} P_{S_{2}^{ \perp}} \rVert_{1}.
	\end{split}
	\end{align}
	The last term on the right-hand side of \eqref{keyIneq2} can be handled by setting $B = \widehat{ \Sigma}_{n} - \Sigma - P_{S_{1}^{ \perp}}( \widehat{ \Sigma}_{n} - \Sigma)P_{S_{2}^{ \perp}}$. The Cauchy--Schwarz inequality and trace duality then delivers the estimate:
	\begin{align} \label{intermediateStep}
	\begin{split}
	\langle \widehat{ \Sigma}_{n} - \Sigma, \widehat{ \Sigma}_{n}^{ \lambda} - A \rangle &\leq \lvert \langle B, \widehat{ \Sigma}_{n}^{ \lambda} - A \rangle \rvert + \lvert \langle \widehat{ \Sigma}_{n} - \Sigma, P_{S_{1}^{ \perp}} \widehat{ \Sigma}_{n}^{ \lambda} P_{S_{2}^{ \perp}} \rangle \rvert \\[0.10cm]
	&\leq \lVert B \rVert_{2} \lVert \widehat{ \Sigma}_{n}^{ \lambda} - A \rVert_{2} + \frac{ \lambda}{2} \lVert P_{S_{1}^{ \perp}} \widehat{ \Sigma}_{n}^{ \lambda} P_{S_{2}^{ \perp}} \rVert_{1}.
	\end{split}
	\end{align}
	For any given matrix $M \in \mathbb{R}^{d \times d}$:
	\begin{equation*}
	M - P_{S_{1}^{ \perp}} M P_{S_{2}^{ \perp}} = P_{S_{1}}M +P_{S_{1}^{ \perp}} M \big(I_{d} - P_{S_{2}^{ \perp}} \big) = P_{S_{1}}M + P_{S_{1}^{ \perp}} M P_{S_{2}}.
	\end{equation*}
	Thus,
	\begin{equation*}
	\lVert M - P_{S_{1}^{ \perp}} M P_{S_{2}^{ \perp}} \rVert_{2} \leq \sqrt{r} \big( \lVert P_{S_{1}} M \rVert_{ \infty} + \lVert P_{S_{1}^{ \perp}} M P_{S_{2}^{ \perp}} \rVert_{ \infty} \big) \leq 2 \sqrt{r} \lVert M \rVert_{ \infty}.
	\end{equation*}
	This shows that $\lVert B \rVert_{2} \leq \lambda \sqrt{r}$. Hence, by combining \eqref{keyIneq2} -- \eqref{intermediateStep} we get
	\begin{equation*}
	\lVert \widehat{ \Sigma}_{n}^{ \lambda} - \Sigma \rVert_{2}^{2} + \lVert \widehat{ \Sigma}_{n}^{ \lambda} - A \rVert_{2}^{2} \leq \lVert A - \Sigma \rVert_{2}^{2} + 3 \lambda \sqrt{r} \lVert \widehat{ \Sigma}_{n}^{ \lambda} - A \rVert_{2}.
	\end{equation*}
	By subtracting $\lVert \widehat{ \Sigma}_{n}^{ \lambda} - A \rVert_{2}^{2}$ on both sides and using the inequality $\alpha^{2}/4 \geq \alpha \beta - \beta^{2}$ with $\alpha = 3 \lambda \sqrt{r}$ and $\beta = \lVert \widehat{ \Sigma}_{n}^{ \lambda} - A \rVert_{2}$, we arrive at the second term in the minimum (with $A = \Sigma$):
	\begin{equation*}
	\lVert \widehat{ \Sigma}_{n}^{ \lambda} - \Sigma \rVert_{2}^{2} \leq \lVert A - \Sigma \rVert_{2}^{2} + 3 \lambda^{2}r.
	\end{equation*} \qed
\end{proof}

\begin{proof}[Proof of Theorem \ref{theorem:bernstein}]
	We set $S_{n} = \sum_{k=1}^{n} X_{k}$ and $\nu_{n} = \sum_{k=1}^{n} C_{k}$. The subadditivity of the maximal eigenvalue function $\lambda_{ \max} ( \cdot)$ implies that
	\begin{equation} \label{keyTerm}
	\lambda_{ \max} \bigg( \theta S_{n} - \frac{ \theta^{2}}{2(1- \theta)} \nu_{n} I_{d} \bigg) \geq \theta x - \frac{ \theta^{2}}{2(1- \theta)} \nu,
	\end{equation}
	whenever $\lambda_{ \max}(S_{n}) \geq x$, $\nu_{n} \leq \nu$ and $\theta \in (0,1)$. By exponentiating both sides of \eqref{keyTerm}, applying the spectral mapping theorem and that $\lambda_{ \max}(A) \leq \tr(A)$ for any positive semidefinite matrix $A$,
	\begin{align}\label{bernsteinProb}
	\begin{split}
	\mathbb{P} \big( \lambda_{ \max}(S_{n}) \geq x, \, \nu_{n} \leq \nu \big) &\leq \mathbb{P} \bigg(Y_{n} \geq \exp \bigg( \theta x - \frac{ \theta^{2}}{2(1- \theta)} \nu \bigg) \bigg) \\[0.10cm]
	&\leq \mathbb{E}[Y_{n}] \exp \bigg(- \theta x + \frac{ \theta^{2}}{2(1- \theta)} \nu \bigg),
	\end{split}
	\end{align}
	where $\displaystyle Y_{n} = \tr \bigg[ \exp \bigg( \theta S_{n} - \frac{ \theta^{2}}{2(1- \theta)} \nu_{n} I_{d} \bigg) \bigg]$. Suppose for the moment that $R=1$, so that $\displaystyle \lVert \mathbb{E}[ X_{k}^{p} \mid \mathcal{G}_{k-1}] \rVert_{ \infty} \leq \frac{p!}{2} C_{k}$ for all $k$ and integers $p \geq 2$. Since $\mathbb{E} \bigl[X_{k} \mid \mathcal{G}_{k-1} \bigr] = 0$, it then follows that
	\begin{align*}
	\bigl\lVert \mathbb{E} \bigl[ \exp( \theta X_{k}) \mid \mathcal{G}_{k-1} \bigr] \bigr\rVert_{ \infty} &\leq 1 + \sum_{p=2}^{ \infty} \frac{ \theta^{p} \lVert \mathbb{E}[ X_{k}^{p} \mid \mathcal{G}_{k-1}] \rVert_{ \infty}}{p!} \\[0.10cm]
	&\leq 1 + \frac{ \theta^{2}}{2(1- \theta)}C_{k} \leq \exp \bigg( \frac{ \theta^{2}}{2(1- \theta)}C_{k} \bigg).
	\end{align*}
	Thus, the matrix
	\begin{equation*}
	\exp \bigg( \frac{ \theta^{2}}{2(1- \theta)}C_{k} \bigg)I_{d}- \mathbb{E} \big[ \exp( \theta X_{k}) \mid \mathcal{G}_{k-1} \big]
	\end{equation*}
	is positive semidefinite, and hence it follows by \citet*[][Lemma 2.1]{tropp:11a} that $(Y_{k})_{k=1}^{n}$ is a positive supermartingale with $Y_{0} = d$. By combining this observation with \eqref{bernsteinProb} and choosing $\theta = x / (x+ \nu)$, we conclude that
	\begin{equation*}
	\mathbb{P} \big( \lambda_{ \max}(S_{n}) \geq x, \, \nu_{n} \leq \nu) \leq d \exp \bigg(- \frac{x^{2}}{2(x+ \nu)} \bigg).
	\end{equation*}
	The general result can now be deduced by taking $\widetilde{X}_{k} = X_{k}/R$. \qed
\end{proof}

A key ingredient to verify the conditions of Theorem \ref{theorem:bernstein} and prove Theorem \ref{applyingBernstein} below is a suitable version of the Burkholder--Davis--Gundy inequality. We shall employ the one given in \cite[Lemma 2.2]{seidler-sobukawa:03a}, which is restated here for convenience. Note that, although their result applies to martingales with values in a general Hilbert space, we specialize the formulation here to the finite-dimensional setting.
\begin{lemma}[\citet{seidler-sobukawa:03a}]\label{bdg} Let $m \in \mathbb{N}$ and $T \in (0, \infty)$. For any constant $p \in [2, \infty)$ there exists $\gamma_{p} \in (0, \infty)$ such that
	\begin{equation*}
	\mathbb{E} \bigg[ \sup_{t \in [0,T]} \Big\lVert \int_{0}^{t} \varphi_{s} \mathrm{d} W_{s} \Big\rVert_{2}^{p} \bigg] \leq \gamma_{p}^{p} \mathbb{E} \bigg[ \Big( \int_{0}^{T} \lVert \varphi_{t} \rVert_{2}^{2} \mathrm{d}t \Big)^{p/2} \bigg]
	\end{equation*}
	for all predictable processes $\varphi \colon \Omega \times [0,T] \to \mathbb{R}^{m \times d}$ with $\int_{0}^{T} \mathbb{E} \big[ \lVert \varphi_{t} \rVert_{2}^{p} \big] \mathrm{d} t < \infty$. Moreover, one can always take
	\begin{equation} \label{CpConstant}
	\gamma_{p} = \frac{4p}{p-1} \sqrt{p+ \frac{1}{2}}.
	\end{equation}
\end{lemma}

With the choice in \eqref{CpConstant}, $\gamma_{p}$ is of smaller asymptotic order (as $p \to \infty$) than the constant associated with the usual Burkholder--Davis--Gundy inequalities available by application of It\^{o}'s formula (cf. the proof of Proposition~2.1 in \cite{marinelli-rockner:16a}).

\begin{proof}[Proof of Theorem \ref{applyingBernstein}]
	We decompose the log-price into the drift and volatility component $Y_{t} = Y_{t}^{ \mu} + Y_{t}^{ \sigma}$, where $Y_{t}^{ \mu} = \int_{0}^{t} \mu_{s} \mathrm{d}s$ and $Y_{t}^{ \sigma} = \int_{0}^{t} \sigma_{s} \mathrm{d} W_{s}$. By the triangle inequality,
	\begin{align} \label{equation:decomposition}
	\begin{split}
	\mathbb{P} \big( \lVert \widehat{ \Sigma}_{n} - \Sigma \rVert_{ \infty} > x \big) &\leq \mathbb{P} \Bigg( \Big\lVert \sum_{k=1}^{ \lfloor \Delta_{n}^{-1} \rfloor}( \Delta_{k}^{n} Y^{ \mu})( \Delta_{k}^{n} Y^{ \mu})^{ \top} \Big\rVert_{ \infty} > \frac{x}{5} \Bigg) \\[0.10cm]
	&+ 2 \mathbb{P} \Bigg( \Big\lVert \sum_{k=1}^{ \lfloor \Delta_{n}^{-1} \rfloor}( \Delta_{k}^{n} Y^{ \mu})( \Delta_{k}^{n} Y^{ \sigma})^{ \top} \Big\rVert_{ \infty} > \frac{x}{5} \Bigg) \\[0.10cm]
	&+ \mathbb{P} \Bigg( \Big\lVert \sum_{k=1}^{ \lfloor \Delta_{n}^{-1} \rfloor}( \Delta_{k}^{n} Y^{ \sigma})( \Delta_{k}^{n} Y^{ \sigma})^{ \top} - \int_{0}^{ \lfloor \Delta_{n}^{-1} \rfloor \Delta_{n}} c_{t} \mathrm{d}t \Big\rVert_{ \infty} > \frac{x}{5} \Bigg) \\[0.10cm]
	&+  \mathbb{P} \Bigg( \Big\lVert \int_{ \lfloor \Delta_{n}^{-1} \rfloor \Delta_{n}}^{1} c_{t} \mathrm{d}t \Big\rVert_{ \infty} > \frac{x}{5} \Bigg) \\[0.10cm]
	&\equiv a_{1} + 2a_{2} + a_{3} + a_{4}.
	\end{split}
	\end{align}
	Jensen's inequality implies	that for the first term in \eqref{equation:decomposition},
	\begin{equation} \label{bound2}
	a_{1} \leq \mathbb{P} \Bigg( \sum_{k=1}^{ \lfloor \Delta_{n}^{-1} \rfloor} \lVert \Delta_{k}^{n} Y^{ \mu} \rVert_{2}^{2} > \frac{x}{5} \Bigg) \leq \mathds{1}_{ \{x < 5 \nu_{ \mu} \Delta_{n} \}}.
	\end{equation}
	The fourth term can be bounded as:
	\begin{equation}\label{bound4}
	a_{4} \leq \mathds{1}_{ \{x < 5 \nu_{c, \infty} \Delta_{n} \}}.
	\end{equation}
	To handle the third term in \eqref{equation:decomposition}, we note it has the form $a_{3} = \sum_{k=1}^{ \lfloor \Delta_{n}^{-1} \rfloor} (A_{k} - B_{k})$, where
	\begin{equation*}
	A_{k} = ( \Delta_{k}^{n} Y^{ \sigma})( \Delta_{k}^{n} Y^{ \sigma})^{ \top}, \quad B_{k} = \int_{(k-1) \Delta_{n}}^{k \Delta_{n}} c_{t} \mathrm{d}t,
	\end{equation*}
	and that $(A_{k} - B_{k})_{k=1}^{ \lfloor \Delta_{n}^{-1} \rfloor}$ is a martingale difference sequence with respect to the filtration $(\mathcal{F}_{k \Delta_{n}})_{k=0}^{ \lfloor \Delta_{n}^{-1} \rfloor}$. Thus, we can apply Theorem \ref{theorem:bernstein} if we can control
	\begin{equation*}
	M_{p,k} \equiv \big\lVert \mathbb{E}[(A_{k} - B_{k})^{p} \mid \mathcal{F}_{(k-1) \Delta_{n}}] \big\rVert_{ \infty},
	\end{equation*}
	for integer $p \geq 2$. First note that, since $\lVert B_{k} \rVert_{ \infty} \leq \nu_{c, \infty} \Delta_{n}$, the sub-multiplicativity of $\lVert \: \cdot \: \rVert_\infty$ and the binomial theorem imply that
	\begin{equation} \label{b1}
	M_{p,k} \leq \big\lVert \mathbb{E}[A_{k}^{p} \mid \mathcal{F}_{(k-1) \Delta_{n}}] \big\rVert_{ \infty} + \sum_{i=1}^{p} \binom{p}{i} \mathbb{E} \big[ \lVert A_{k} \rVert_{ \infty}^{p-i} \mid \mathcal{F}_{(k-1) \Delta_{n}} \big]( \nu_{c, \infty} \Delta_{n})^{i}.
	\end{equation}
	We start by analyzing the second term on the right-hand side of \eqref{b1}. For an arbitrary integer $m \geq 1$ and any $\mathcal{F}_{(k-1) \Delta_{n}}$-measurable set $A$, it follows from Lemma~\ref{bdg} that
	\begin{equation*}
	\mathbb{E} \big[ \lVert A_{k} \rVert_{ \infty}^{m} \mathds{1}_{A} \big] = \mathbb{E} \bigg[ \Big\lVert \int_{(k-1) \Delta_{n}}^{k \Delta_{n}} \sigma_{t} \mathds{1}_{A} \mathrm{d}W_{t} \Big\rVert_{2}^{2m} \bigg] \leq \gamma_{2m}^{2m} \mathbb{E} \bigg[ \Big( \int_{(k-1) \Delta_{n}}^{k \Delta_{n}} \tr(c_{t}) \mathrm{d}t \Big)^{m} \mathds{1}_{A} \bigg]
	\end{equation*}
	and, thus,
	\begin{equation} \label{boundA}
	\mathbb{E}[ \lVert A_{k} \rVert_\infty^m \mid \mathcal{F}_{(k-1)\Delta_n}]^{1/m} \leq \alpha m  \nu_{c,2}\Delta_n
	\end{equation}
	for a suitable absolute constant $\alpha \geq 1$. In particular, \eqref{boundA} shows that $\mathbb{E} \big[ \lVert A_{k} \rVert_{ \infty}^{p-i} \mid \mathcal{F}_{(k-1) \Delta_{n}}] \leq ( \alpha p \nu_{c,2} \Delta_{n})^{p-i}$ for $i = 1, \dots, p$. Together with the mean value theorem and Stirling's formula, we can therefore establish that
	\begin{align} \label{firstTerm}
	\begin{split}
	\sum_{i=1}^p \binom{p}{i} \mathbb{E} \big[ \lVert A_{k} \rVert_{ \infty}^{p-i} \mid \mathcal{F}_{(k-1) \Delta_{n}} \big]( \nu_{c, \infty} \Delta_{n})^{i} &\leq ( \alpha p \nu_{c,2} \Delta_{n})^{p} \sum_{i=1}^{p} \binom{p}{i} \bigg( \frac{ \nu_{c, \infty}}{ \alpha p \nu_{c,2}} \bigg)^{i} \\[0.10cm]
	&\leq \frac{p!}{2}(e \alpha \nu_{c,2} \Delta_{n})^{p} \bigg[ \Big(1+ \frac{ \nu_{c, \infty}}{ \alpha p \nu_{c,2}} \Big)^{p} - 1 \bigg] \\[0.10cm]
	&\leq \frac{p!}{2}(2e \alpha \nu_{c,2} \Delta_{n})^{p} \frac{ \nu_{c, \infty}}{ \nu_{c,2}}.
	\end{split}
	\end{align}
	Now, look at the first term on the right-hand side of \eqref{b1}. For a fixed $u\in \mathbb{R}^d$ with $\lVert u \rVert_2 = 1$, it follows from the Cauchy--Schwarz inequality that
	\begin{align} \label{b2}
	\begin{split}
	u^{ \top} \mathbb{E} \big[A_{k}^{p} \mid \mathcal{F}_{(k-1) \Delta_{n}} \big] u &= \mathbb{E} \big[ \lVert A_{k} \rVert_{ \infty}^{p-1} (u^{ \top}( \Delta_{k}^{n} Y))^{2} \mid \mathcal{F}_{(k-1) \Delta_{n}} \big] \\[0.10cm]
	&\leq \mathbb{E} \big[ \lVert A_{k} \rVert_{ \infty}^{2(p-1)} \mid \mathcal{F}_{(k-1) \Delta_{n}} \big]^{1/2} \mathbb{E} \big[(u^{ \top} \Delta_{k}^{n} Y)^{4} \mid \mathcal{F}_{(k-1) \Delta_{n}} \bigr]^{1/2}.
	\end{split}
	\end{align}
	The first term in \eqref{b2} is handled by \eqref{boundA} with $m=2(p-1)$:
	\begin{equation} \label{b21}
	\mathbb{E} \big[ \lVert A_{k} \rVert_{ \infty}^{2(p-1)} \mid \mathcal{F}_{(k-1) \Delta_{n}} \big]^{1/2} \leq (2 \alpha(p-1) \nu_{c,2} \Delta_{n})^{p-1} \leq \frac{p!}{2} \frac{(2e \alpha \nu_{c,2} \Delta_{n})^p}{ \nu_{c,2} \Delta_{n}}.
	\end{equation}
	The second term is treated as:
	\begin{align} \label{b22}
	\begin{split}
	\mathbb{E} \big[(u^{ \top} \Delta_{k}^{n}Y)^{4} \mid \mathcal{F}_{(k-1) \Delta_{n}} \big]^{1/2} &\leq \gamma_{4}^{2} \mathbb{E} \bigg[ \Big( \int_{(k-1) \Delta_{n}}^{k \Delta_{n}} u^{ \top} c_{t} u \mathrm{d} \Big)^{2} \mid \mathcal{F}_{(k-1) \Delta_{n}} \bigg]^{1/2} \\[0.10cm]
	&\leq \gamma_{4}^{2} \nu_{c, \infty} \Delta_{n},
	\end{split}
	\end{align}
	where Lemma~\ref{bdg} is used here for the conditional expectation. Consequently, by combining \eqref{b2} -- \eqref{b22}, we get the estimate
	\begin{equation}\label{specBound}
	u^{ \top} \mathbb{E} \big[A_{k}^{p} \mid \mathcal{F}_{(k-1) \Delta_{n}}] u \leq \gamma_{4}^{2} \frac{p!}{2} (2e \alpha \nu_{c,2} \Delta_{n})^{p} \frac{ \nu_{c, \infty}}{ \nu_{c,2}}.
	\end{equation}
	Since this analysis holds for an arbitrary unit vector $u$, the right-hand side of \eqref{specBound} is an upper bound for $\big\lVert \mathbb{E}[A_{k}^{p} \mid \mathcal{F}_{(k-1) \Delta_{n}}] \big\rVert_{ \infty}$. This fact, together with \eqref{b1} and \eqref{firstTerm}, shows that
	\begin{equation}\label{momentBound}
	M_{p,k} \leq  \frac{p!}{2}( \tilde{ \alpha} \nu_{c,2} \Delta_{n})^{p} \frac{ \nu_{c, \infty}}{ \nu_{c,2}}
	\end{equation}
	for a suitable absolute constant $\tilde{\alpha} \geq 2$. Now, it follows that Theorem \ref{theorem:bernstein} is applicable with $R = \tilde{ \alpha} \nu_{c,2} \Delta_{n}$ and $C_{k} = C_{1} = \nu_{c, \infty}/ \nu_{c,2}$. Since
	\begin{equation*}
	\nu R =  \tilde{ \alpha} \nu_{c, \infty} \lfloor \Delta_{n}^{-1} \rfloor \Delta_{n} \geq \tilde{ \alpha} \nu_{c, \infty} (1- \Delta_{n}) \geq \nu_{c, \infty}
	\end{equation*}
	and $4 \nu R^{2} \leq 4 \tilde{ \alpha}^{2} \nu_{c,2} \nu_{c, \infty} \Delta_{n}$, the inequality \eqref{equation:subexponential} implies that
	\begin{equation}\label{noDrift}
	\mathbb{P} \Bigg( \Big\lVert \sum_{k=1}^{ \lfloor \Delta_{n}^{-1} \rfloor}(A_{k}-B_{k}) \Big\rVert_{ \infty} > y \Bigg) \leq \tau (y) \equiv 2d \exp \bigg(- \frac{y}{4 \tilde{ \alpha}^{2} \nu_{c,2} \nu_{c, \infty} \Delta_{n}} \min \{y , \nu_{c, \infty} \} \bigg).
	\end{equation}
	In particular, \eqref{noDrift} shows that
	\begin{equation}\label{bound1}
	a_{3} \leq \tau \Big( \frac{x}{5} \Big).
	\end{equation}
	Finally, for the second term in \eqref{equation:decomposition} we invoke the Cauchy--Schwarz inequality and consider an $x$ for which $\frac{x^{2}}{25 \nu_{ \mu} \Delta_{n}} - \nu_{c,2} \geq \frac{ \nu_{c,2}}{5 \nu_{c, \infty}}x$ to deduce the following initial bound:
	\begin{align} \label{secondTerm}
	\begin{split}
	a_{2} &\leq \mathbb{P} \bigg( \Big( \sum_{k=1}^{ \lfloor \Delta_{n}^{-1} \rfloor} \lVert \Delta_{k}^{n} Y^{ \mu} \rVert_{2}^{2} \Big) \Big( \sum_{k=1}^{ \lfloor \Delta_{n}^{-1} \rfloor} \lVert \Delta_{k}^{n} Y^{ \sigma} \rVert_{2}^{2} \Big) > \Big( \frac{x}{5} \Big)^{2} \bigg) \\[0.10cm]
	&\leq \mathbb{P} \bigg( \sum_{k=1}^{ \lfloor \Delta_{n}^{-1} \rfloor} \lVert \Delta_{k}^{n} Y^{ \sigma} \rVert_{2}^{2} > \frac{x^{2}}{25 \nu_{ \mu} \Delta_{n}} \bigg) \\[0.10cm]
	&\leq \mathbb{P} \bigg( \sum_{k=1}^{ \lfloor \Delta_{n}^{-1} \rfloor} \Big( \lVert \Delta_{k}^{n} Y^{ \sigma} \rVert_{2}^{2} - \int_{(k-1) \Delta_{n}}^{k \Delta_{n}} \tr(c_{t}) \mathrm{d}t \Big) > \frac{x^{2}}{25 \nu_{ \mu} \Delta_{n}} - \nu_{c,2} \bigg) \\[0.10cm]
	&\leq \mathbb{P} \bigg( \sum_{k=1}^{ \lfloor \Delta_{n}^{-1} \rfloor} \Big( \lVert \Delta_{k}^{n} Y^{ \sigma} \rVert_{2}^{2} - \int_{(k-1) \Delta_{n}}^{k \Delta_{n}} \tr(c_{t}) \mathrm{d}t \Big) > \frac{ \nu_{c,2}}{5 \nu_{c, \infty}}x \bigg).
	\end{split}
	\end{align}
	Next, observe that
	\begin{equation*}
	X_{k} = \lVert \Delta_{k}^{n} Y^{ \sigma} \rVert_{2}^{2} - \int_{(k-1) \Delta_{n}}^{k \Delta_{n}} \tr(c_{t}) \mathrm{d}t, \qquad k = 1, \dots, \lfloor \Delta_{n}^{-1} \rfloor,
	\end{equation*}
	is a martingale difference sequence. By \eqref{boundA}, we see that
	\begin{equation*}
	\mathbb{E} \big[ \vert X_{k} \vert^{p} \mid \mathcal{F}_{(k-1) \Delta_{n}} \big] \leq \frac{p!}{2}( \alpha \nu_{c,2} \Delta_{n})^{p},
	\end{equation*}
	for a suitable absolute constant $\alpha \geq 2$, and thus \eqref{equation:subexponential} ensures that
	\begin{equation*}
	\mathbb{P} \bigg( \Bigl\vert \sum_{k=1}^{n} X_{k} \Big\vert > y \bigg) \leq \tau \Big( \frac{ \nu_{c, \infty}}{ \nu_{c,2}}y \Big),
	\end{equation*}
	where we assume $\tilde{ \alpha}$ in $\tau$ is larger than $\alpha$. This estimate together with \eqref{secondTerm} means that
	\begin{equation}\label{bound3}
	a_{2} \leq \tau \Big( \frac{x}{5} \Big).
	\end{equation}
	By combining \eqref{bound2}, \eqref{bound4}, \eqref{bound1}, and \eqref{bound3} we conclude that
	\begin{equation} \label{finalInequality}
	\mathbb{P} \big( \lVert \widehat{ \Sigma}_{n} - \Sigma \rVert_{ \infty} > x \big) \leq 3 \tau \Big( \frac{x}{5} \Big) \leq 6d \exp \Big\{- \frac{x}{100 \tilde{ \alpha}^{2} \nu_{c,2} \nu_{c, \infty} \Delta_{n}} \min \{x , \nu_{c, \infty} \} \Big\}
	\end{equation}
	if $x \geq 5 \Delta_{n} \max\{ \nu_{ \mu}, \nu_{c, \infty} \}$ and $\frac{x^{2}}{25 \nu_{ \mu} \Delta_{n}} - \nu_{c,2} \geq \frac{ \nu_{c,2}}{5 \nu_{c, \infty}}x$. In particular, \eqref{finalInequality} holds whenever $x \geq \gamma \max \big\{ \frac{ \nu_{c,2} \nu_{ \mu} \Delta_{n}}{ \nu_{c, \infty}}, \sqrt{ \nu_{c,2} \nu_{ \mu} \Delta_{n}}, \nu_{c, \infty} \Delta_{n} \big\}$, where $\gamma$ is a sufficiently large absolute constant. \qed
\end{proof}

Note that, in the above calculations, one has to be careful in order to achieve optimal rates. Indeed, instead of the estimate \eqref{momentBound}, it might have been more natural to rely on the following (seemingly harmless) sequence of inequalities using \eqref{boundA} and Stirling's formula:
\begin{equation}\label{momentBound2}
M_{p,k} \leq  2^p\Bigl(\mathbb{E}\bigl[\lVert A_k\rVert_\infty^p\mid \mathcal{F}_{(k-1)\Delta_n}\bigr] + \mathbb{E}\bigl[\lVert B_k \rVert_\infty^p\mid \mathcal{F}_{(k-1)\Delta_n} \bigr]\Bigr)\leq p! (2e\alpha \nu_{c,2}\Delta_n)^p.
\end{equation}
However, with notation as in Theorem \ref{theorem:bernstein}, this would result in $\nu R \asymp \nu_{c,2}$, while \eqref{momentBound} satisfies $\nu R \asymp \nu_{c,\infty}$. Since $\nu R$ determines the rate and $\nu_{c,\infty}$ can be of strictly smaller order than $\nu_{c,2}$ (see the discussion in the end of Section \ref{section:concentration}), this implies that a concentration inequality based on \eqref{momentBound2} cannot be optimal.

\begin{proof}[Proof of Theorem \ref{theorem:concentration-prv}]
	Consider a fixed $\tau \in (0,\infty)$. Theorem \ref{applyingBernstein} implies the existence of an absolute constant $\alpha$ such that
	\begin{equation*}
	\mathbb{P}( 2\lVert \widehat{\Sigma}_n - \Sigma \rVert_\infty > \lambda) \leq 6d \exp \Bigl(-\frac{\lambda}{\alpha \nu_{c,2} \nu_{c,\infty} \Delta_n} \min\{\lambda,  \nu_{c,\infty}\} \Bigr)
	\end{equation*}
	as long as $\lambda \geq \alpha \max\bigl\{\frac{\nu_{c,2}\nu_{\mu} \Delta_n}{\nu_{c,\infty}}, \sqrt{\nu_{c,2}\nu_\mu \Delta_n}, \nu_{c,\infty}\Delta_n \bigr\}$. Thus, under this restriction on $\lambda$ it follows that $2\lVert \widehat{\Sigma}_n -\Sigma \rVert_\infty \leq \lambda$ with probability at least $1-e^{-\tau}$ if
	\begin{equation*}
	6d \exp \Bigl(-\frac{\lambda}{\alpha \nu_{c,2} \nu_{c,\infty} \Delta_n} \min\{\lambda, \nu_{c,\infty}\} \Bigr) \leq \exp(-\tau)
	\end{equation*}
	or, in particular, if
	\begin{equation}\label{ToShowRank}
	\lambda \geq  \gamma \max\Bigl\{\sqrt{ \nu_{c,2} \nu_{c,\infty} \Delta_n (\log(d) + \tau)}, \nu_{c,2}\Delta_n (\log(d)+ \tau),  \frac{\nu_{c,2}\nu_\mu \Delta_n}{\nu_{c,\infty}}\Bigr\}
	\end{equation}
	for a suitably chosen absolute constant $\gamma$. In view of Proposition \ref{boundOn}, the proof is complete. \qed
\end{proof}

\begin{proof}[Proof of Theorem \ref{lowerBound}]
	Under $\mathbb{P}_A$, we have that $Z_k = \Delta_k^n Y/\sqrt{\Delta_n}$, $k=1,\dots, \lfloor \Delta_n^{-1}\rfloor$, are i.i.d. Gaussian random vectors with covariance matrix $A$. Consequently, since $\lfloor \Delta_n^{-1}\rfloor \geq r^2$, \cite[Theorem 2]{lounici:14a} implies that
	\begin{equation}\label{lowerBounds1}
	\inf_{\widetilde{\Sigma}}\sup_{A\in \mathcal{C}_r} \mathbb{P}_A\Bigl(\lVert \widetilde{\Sigma}-A \rVert_2^2> \underline{\gamma} \frac{\lVert A \rVert_\infty^2 r_e(A)\rank (A)}{\lfloor \Delta_n^{-1}\rfloor}\Bigr) \geq \beta
	\end{equation}
	for suitable absolute constants $\beta \in (0,1)$ and $\underline{\gamma} \in (0,\infty)$, 
	and this proves the result. \qed
\end{proof}

\begin{proof}[Proof of Theorem \ref{theorem:rank-relation}]
	Set $\hat{r}= \rank (\widehat{\Sigma}^\lambda_n)$. To show the upper bound in \eqref{rank1} it suffices to argue that $\Sigma$ has at least $\hat{r}$ singular values which are larger than $(\lambda-\bar{\lambda})/2$ or, equivalently, $s_{\hat{r}}(\Sigma)\geq (\lambda - \bar{\lambda})/2$. Observe first that the representation \eqref{RVthreshold} implies $s_{\hat{r}}(\widehat{\Sigma}_n)>\lambda/2$ (recall that $s_k(A)$ refers to the $k$th largest singular value of $A$). By using this inequality together with the fact that $A\mapsto s_k(A)$ is Lipschitz continuous with constant $1$ with respect to the spectral norm, we establish that
	\begin{equation*}
	s_{\hat{r}} (\Sigma) \geq s_{\hat{r}} (\widehat{\Sigma}_n)- \vert s_{\hat{r}} (\widehat{\Sigma}_n)-s_{\hat{r}} (\Sigma)\vert \geq \frac{\lambda-\bar{\lambda}}{2}.
	\end{equation*}
	In order to prove the lower bound of \eqref{rank1}, we need to show that $s_{r_\lambda}(\widehat{\Sigma}^\lambda_n) >0$ with $r_\lambda= \rank (\Sigma;\lambda)$. However, this follows immediately from the following sequence of inequalities:
	\begin{equation*}
	s_{r_\lambda} (\widehat{\Sigma}^\lambda_n) \geq s_{r_\lambda}(\widehat{\Sigma}_n)-\frac{\lambda}{2} \geq s_{r_\lambda} (\Sigma)-\lVert \widehat{\Sigma}_n-\Sigma \rVert_\infty - \frac{\lambda}{2} \geq \frac{\lambda-\bar{\lambda}}{2}>0.
	\end{equation*}
	The last part of the result is a direct consequence of the inequality \eqref{rank1}, since $\rank(\Sigma;x) = \rank (\Sigma)$ for $x \in (0,s]$. This finishes the proof. \qed
\end{proof}

\begin{proof}[Proof of Corollary \ref{rankConcentration}]
	For an arbitrary number $\bar{\lambda}\in (0,\infty)$, Proposition \ref{boundOn} and Theorem \ref{theorem:rank-relation} imply that both
	\begin{equation}\label{performance:rank}
	\lVert\widehat{\Sigma}^\lambda_n -\Sigma \rVert_2^2 \leq 3 \lambda^2 \rank (\Sigma)
	\end{equation}
	and
	\begin{equation}\label{performance:rank2}
	\rank (\Sigma;\lambda)\leq\rank (\widehat{\Sigma}^\lambda_n) \leq \rank (\Sigma;\tfrac{1}{2}(\lambda - \bar{\lambda}))
	\end{equation}
	on the set $A(\bar{\lambda})\equiv \{2\lVert \widehat{\Sigma}_n-\Sigma \rVert_\infty \leq \bar{\lambda}\}$ whenever $\lambda > \bar{\lambda}$. By the proof of Theorem \ref{theorem:concentration-prv}, in particular \eqref{ToShowRank} together with the fact that $\nu_\mu \leq \nu_{c,\infty}$, it follows that for any $\tau \in (0,\infty)$, the set $A(\bar{\lambda})$ has probability at least $1-\exp(-\tau)$ if
	\begin{equation}\label{rankstep}
	\bar{\lambda} =  \bar{\gamma} \max\Bigl\{\sqrt{ \nu_{c,2} \nu_{c,\infty} \Delta_n (\log(d) + \tau)}, \nu_{c,2}\Delta_n (\log(d)+ \tau)\Bigr\}
	\end{equation}
	for a sufficiently large absolute constant $\bar{\gamma}$. By considering $\tau=\log(d)$ and using that $\Delta_n^{-1}\geq 2 \frac{\nu_{c,2}}{\nu_{c,\infty}} \log(d)$, the maximum in \eqref{rankstep} equals its first term. With this choice of $\tau$, the set $A(\bar{\lambda})$ has probability at least $1-d^{-1}$, and $(\lambda-\bar{\lambda})/2\geq \delta \lambda$ when $\lambda$ is given by \eqref{rankRegularization} as long as we choose $\gamma \geq \bar{\gamma}\sqrt{2}/(1-2\delta)$. Consequently, by plugging the specific values of $\lambda$ and $\bar{\lambda}$ into \eqref{performance:rank} and \eqref{performance:rank2}, we have established the first part of the result. The last part of the result follows immediately from Theorem \ref{theorem:rank-relation}. \qed
\end{proof}

\begin{proof}[Proof of Theorem \ref{localVol}]
	First, let $\delta_n\in \{0,1\}$ be defined such that
	\begin{equation}\label{deltaN}
	\lfloor (t+h_n)\Delta_n^{-1} \rfloor - \lfloor t \Delta_n^{-1}\rfloor = \lfloor (h_n+\delta_n\Delta_n)\Delta_n^{-1} \rfloor,
	\end{equation}
	and set $\tau (s) = \varepsilon s + \lfloor t \Delta_n^{-1}\rfloor \Delta_n$ with $\varepsilon = h_n + \delta_n \Delta_n$. Clearly, $\bar{Y}_s = Y_{\tau (s)}$, $s\in [0,1]$, is a diffusion with drift $\bar{\mu}_s = \varepsilon \mu_{\tau (s)}$, volatility $\bar{\sigma}_s = \sqrt{\varepsilon} \sigma_{\tau (s)}$, and driving (standard) Brownian motion $\bar{W}_s = (W_{\tau (s)}-W_{\tau (0)})/\sqrt{\varepsilon}$. Moreover, its RV at time $1$ based on a sampling frequency of $\bar{\Delta}_n = \Delta_n/\varepsilon$ coincides with $\widehat{\Sigma}_n(t;t+h_n)$:
	\begin{equation*}
	\widehat{\Sigma}_n(t;t+h_n) = \sum_{k=1}^{\lfloor \bar{\Delta}^{-1}_n \rfloor} (\bar{\Delta}^n_k \bar{Y})(\bar{\Delta}^n_k \bar{Y})^\top,\qquad \bar{\Delta}^n_k \bar{Y} = \bar{Y}_{k\bar{\Delta}_n}-\bar{Y}_{(k-1)\bar{\Delta}_n}.
	\end{equation*}
	Hence, the first step is to apply Corollary \ref{consequenceOfPenConcentration} to the process $(\bar{Y}_s)_{s\in [0,1]}$ with sampling frequency $\bar{\Delta}_n$. To this end, note that
	\begin{equation*}
	\sup_{s\leq 1}\lVert \bar{\mu}_s \rVert^2_2\leq \varepsilon \nu_\mu,\quad \sup_{s\leq 1}\tr (\bar{\sigma}_s\bar{\sigma}_s^\top)\leq \varepsilon \nu_{c,2},\quad \text{and}\quad \sup_{s\leq 1}\lVert \bar{\sigma}_s\bar{\sigma}_s^\top\rVert_\infty \leq \varepsilon \nu_{c,\infty}.
	\end{equation*}
	Since we also have $\varepsilon \in [h_n,2h_n]$ and $h_n/\Delta \geq 2\frac{\nu_{c,2}}{\nu_{c,\infty}}\log(d)$, it follows that the assumptions of Corollary \ref{consequenceOfPenConcentration} are satisfied and we deduce that
	\begin{equation}\label{local1}
	\bigl\lVert \widehat{\Sigma}^\lambda_n (t;t+h_n)-\bar{\Sigma} \bigr\rVert_2^2 \leq 3 \gamma^2 \nu_{c,2}\nu_{c,\infty}
	h_n\Delta_n\rank (\bar{\Sigma}) \log(d)
	\end{equation}
	with probability at least $1-d^{-1}$ when $\lambda$ meets \eqref{localLambda} and $\gamma$ is a sufficiently large absolute constant. Here $\bar{\Sigma}$ denotes the QV of $(\bar{Y}_s)_{s\in [0,1]}$ at time $1$. By dividing both sides of the inequality \eqref{local1} by $h_n^2$,
	\begin{equation}\label{localVol:term1}
	\Bigl\lVert \hat{c}_n^\lambda (t) -\frac{\bar{\Sigma}}{h_n} \Bigr\rVert_2^2 \leq 3 \gamma^2\frac{\nu_{c,2}\nu_{c,\infty} \Delta_n \rank (\bar{\Sigma}) \log(d)}{h_n}.
	\end{equation}
	The error $\lVert h_n^{-1}\bar{\Sigma}- \varepsilon^{-1}\bar{\Sigma}\rVert_2$ can be bounded in the following way using that $h_n/\Delta_n\geq \frac{\nu_{c,2}}{2\nu_{c,\infty}}$:
	\begin{equation}\label{localVol:term3}
	\Bigl\lVert \frac{\bar{\Sigma}}{h_n}- \frac{\bar{\Sigma}}{\varepsilon}\Bigr\rVert_2^2 \leq \Bigl(\frac{\nu_{c,2}\Delta_n}{h_n}\Bigr)^2 \leq \frac{2\nu_{c,2}\nu_{c,\infty}\Delta_n}{h_n}.
	\end{equation}
	Now, for any given $\beta \in (1,\infty)$, we want to argue that $\lVert \varepsilon^{-1}\bar{\Sigma}-c_t \rVert_2$ is small with probability $1-\beta^{-1}$. To do so, note initially that $t\in A_n\equiv [\lfloor t \Delta_n^{-1}\rfloor \Delta_n, \lfloor t \Delta_n^{-1}\rfloor \Delta_n+ \varepsilon ]$ and that, for any $s,u\in A_n$, we have
	$\nu_{c,\psi} \geq  \lVert c_s-c_u \rVert_\psi/\sqrt{2h_n}$. Using these two facts, and by relying on Jensen's and Markov's inequality, we do the following computations for an arbitrary $x>0$:
	\begin{align*}
	\mathbb{P}\Bigl(\Bigl\lVert\frac{\bar{\Sigma}}{\varepsilon}- c_t\Bigr\rVert_2>x\Bigr)  &\leq \mathbb{P}\Bigl(\frac{1}{\varepsilon}\int_{A_n} \lVert c_s-c_t\rVert_2 \mathrm{d} s >x\Bigr)\\
	&\leq \mathbb{P}\Bigl(\frac{1}{\varepsilon}\int_{A_n} \psi\Bigl(\frac{\lVert c_s-c_t\rVert_2}{\lVert c_s-c_t \rVert_\psi}\Bigr) \mathrm{d} s >\psi \Bigl(\frac{x}{\nu_{c,\psi}\sqrt{2h_n}}\Bigr)\Bigr).\\
	&\leq \psi \Bigl(\frac{x}{\nu_{c,\psi}\sqrt{2h_n}}\Bigr)^{-1}.
	\end{align*}
	It follows that, with probability at least $1-\beta^{-1}$,
	\begin{equation}\label{localVol:term2}
	\Bigl\lVert \frac{\bar{\Sigma}}{\varepsilon}-c_t \Bigr\rVert_2^2 \leq 2h_n\nu_{c,\psi}^2\psi^{-1}(\beta)^2.
	\end{equation}
	By choosing $\beta = \psi (\sqrt{\log(d)})$ and combining \eqref{localVol:term1} -- \eqref{localVol:term2} we conclude that, with probability at least $1-d^{-1}-\psi (\sqrt{\log(d)})^{-1}$,
	\begin{equation*}
	\lVert \hat{c}^\lambda_n(t) - c_t \rVert_2^2 \leq \kappa \gamma^2 \Bigl(\frac{\nu_{c,2}\nu_{c,\infty}\Delta_n\rank (\bar{\Sigma}) }{h_n} + h_n \nu_{c,\psi}^2\Bigr) \log(d)
	\end{equation*}
	for a suitably chosen absolute constant $\kappa$. Since $A_n \subseteq [t-\Delta_n,t+h_n+\Delta_n]$ and $h_n \geq 2 \Delta_n$, it follows that $\rank (\bar{\Sigma})\leq \rank (\Sigma (t-\tfrac{h_n}{2};t+\tfrac{3h_n}{2}))$, and the proof is complete. \qed
\end{proof}

\begin{proof}[Proof of Theorem \ref{localVol:rank}]
	Similarly to the proof of Theorem \ref{localVol}, we consider the time-changed process $\bar{Y}_s = Y_{\tau (s)}$, $s\in [0,1]$, and apply Corollary \ref{rankConcentration} to deduce that, with probability at least $1-d^{-1}$,
	\begin{equation*}
	\Bigl\lVert \hat{c}_n^\lambda (t) -\frac{\bar{\Sigma}}{h_n} \Bigr\rVert_2^2 \leq 3\gamma^2\frac{\nu_{c,2}\nu_{c,\infty} \Delta_n \rank (\bar{\Sigma}) \log(d)}{h_n}
	\end{equation*}
	and
	\begin{equation}\label{rankProof}
	\rank (\bar{\Sigma};\lambda)\leq \rank (\hat{c}_n^\lambda (t)) \leq \rank (\bar{\Sigma};2\delta \lambda).
	\end{equation}
	Here we recall that $\tau (s) = \varepsilon s + \lfloor t \Delta_n^{-1}\rfloor \Delta_n$ and $\varepsilon = h_n + \delta_n \Delta_n$, where $\delta_n\in \{0,1\}$ is chosen such that \eqref{deltaN} holds. By following the exact same arguments as in the proof of Theorem \ref{localVol} we obtain the estimate
	\begin{equation*}
	\lVert \hat{c}^\lambda_n(t) - c_t \rVert_2^2 \leq \kappa\gamma^2 \Bigl(\frac{\nu_{c,2}\nu_{c,\infty} \Delta_n\rank (\bar{\Sigma}) }{h_n} + h_n \nu_{c,\psi}^2\Bigr) \log(d),
	\end{equation*}
	which applies with probability at least $1-d^{-1}-\psi (\sqrt{\log(d)})^{-1}$ for a suitably chosen absolute constant $\kappa$. The inequality \eqref{rankInequality} follows immediately from the fact that $\rank (\bar{\Sigma})\leq \rank (\Sigma(t-\tfrac{h_n}{2};t+\tfrac{3h_n}{2}))$. To establish the bounds \eqref{rankBounds_theorem} on $\rank (\hat{c}^\lambda_n(t))$ note initially that, from the proof of Theorem \ref{localVol} (particularly, \eqref{localVol:term2}), we may in fact assume that
	\begin{equation*}
	\Bigl\lVert \frac{\bar{\Sigma}}{\varepsilon}-c_t \Bigr\rVert_2^2 \leq 2 h_n \nu_{c,\psi}^2 \log(d)
	\end{equation*}
	on the event that we are considering. Consequently, since singular values are Lipschitz continuous with constant $1$ with respect to $\lVert\: \cdot \: \rVert_2$,
	\begin{equation*}
	\Bigl\vert s_k\Bigl(\frac{\bar{\Sigma}}{\varepsilon} \Bigr)- s_k(c_t) \Bigr\vert \leq \nu_{c,\psi}\sqrt{2h_n\log(d)}
	\end{equation*}
	for $k=1,\dots, d$. By using this observation, the fact that $\varepsilon \in [h_n,2h_n]$, and the explicit expression for $\lambda$ the following two implications can be deduced:
	\begin{align*}
	s_k(\bar{\Sigma}) &\geq 2\delta \lambda\quad &\Longrightarrow\quad s_k (c_t) &\geq \delta \gamma\Bigl(\sqrt{\frac{ \nu_{c,2}\nu_{c,\infty}\Delta_n}{h_n}}-\nu_{c,\psi}\sqrt{h_n}\Bigr)\sqrt{\log(d)},\\
	s_k( \bar{\Sigma}) &< \lambda \quad &\Longrightarrow \quad s_k (c_t) &<\gamma \Bigl(\sqrt{\frac{ \nu_{c,2}\nu_{c,\infty}\Delta_n}{h_n}}+\nu_{c,\psi}\sqrt{h_n}\Bigr)\sqrt{\log(d)}.
	\end{align*}
	We have also imposed the innocent assumption that $\gamma$ is chosen such that $\delta\gamma \geq \sqrt{2}$.
	From these two implications we conclude that $\rank (\bar{\Sigma};2\delta\lambda)\leq \rank (c_t;\underline{\varepsilon})$ and $\rank (\bar{\Sigma};\lambda)\geq \rank (c_t;\overline{\varepsilon})$, which (in view of \eqref{rankProof}) establishes \eqref{rankBounds_theorem}. The last statement in the result is a direct consequence of \eqref{rankBounds_theorem}, and thus the proof is complete. \qed
\end{proof}

\clearpage


\renewcommand{\baselinestretch}{1.0}
\small
\bibliographystyle{rfs}
\bibliography{userref}

@ARTICLE{ait-sahalia-xiu:17a,
 AUTHOR = {Y. A\"{i}t-Sahalia and D. Xiu},
 YEAR = {2017},
 TITLE = {{Using principal component analysis to estimate a high dimensional factor model with high-frequency data}},
 JOURNAL = {Journal of Econometrics},
 VOLUME = {201},
 NUMBER = {2},
 PAGES = {384--399}
}

@ARTICLE{ait-sahalia-xiu:19b,
 AUTHOR = {Y. A\"{i}t-Sahalia and D. Xiu},
 YEAR = {2019},
 TITLE = {{Principal component analysis of high-frequency data}},
 JOURNAL = {Journal of the American Statistical Association},
 VOLUME = {114},
 NUMBER = {525},
 PAGES = {287--303}
}

@ARTICLE{andersen-bollerslev:98a,
 AUTHOR = {T. G. Andersen and T. Bollerslev},
 YEAR = {1998},
 TITLE = {{Answering the skeptics: Yes, standard volatility models do provide accurate forecasts}},
 JOURNAL = {International Economic Review},
 VOLUME = {39},
 NUMBER = {4},
 PAGES = {885--905}
}

@ARTICLE{andersen-bollerslev-diebold-labys:03a,
 AUTHOR = {T. G. Andersen and T. Bollerslev and F. X. Diebold and P. Labys},
 YEAR = {2003},
 TITLE = {{Modeling and forecasting realized volatility}},
 JOURNAL = {Econometrica},
 VOLUME = {71},
 NUMBER = {2},
 PAGES = {579--625}
}

@ARTICLE{argyriou-evgeniou-pontil:08a,
 AUTHOR = {A. Argyriou and T. Evgeniou and M. Pontil},
 YEAR = {2008},
 TITLE = {{Convex multi-task feature learning}},
 JOURNAL = {Machine Learning},
 VOLUME = {73},
 NUMBER = {3},
 PAGES = {243--272}
}

@ARTICLE{bach:08a,
 AUTHOR = {F. R. Bach},
 YEAR = {2008},
 TITLE = {{Consistency of trace norm minimization}},
 JOURNAL = {Journal of Machine Learning Research},
 VOLUME = {9},
 NUMBER = {35},
 PAGES = {1019--1048}
}

@INCOLLECTION{barndorff-nielsen-graversen-jacod-podolskij-shephard:06a,
 AUTHOR = {O. E. Barndorff-Nielsen and S. E. Graversen and J. Jacod and M. Podolskij and N. Shephard},
 YEAR = {2006},
 TITLE = {{A central limit theorem for realized power and bipower variations of continuous semimartingales}},
 BOOKTITLE = {From Stochastic Calculus to Mathematical Finance: The Shiryaev Festschrift},
 EDITOR = {Y. Kabanov and R. Lipster and J. Stoyanov},
 PAGES = {33--68},
 PUBLISHER = {Springer},
 ADDRESS = {Berlin}
}

@ARTICLE{barndorff-nielsen-shephard:02a,
 AUTHOR = {O. E. Barndorff-Nielsen and N. Shephard},
 YEAR = {2002},
 TITLE = {{Econometric analysis of realized volatility and its use in estimating stochastic volatility models}},
 JOURNAL = {Journal of the Royal Statistical Society: Series B},
 VOLUME = {64},
 NUMBER = {2},
 PAGES = {253--280}
}

@ARTICLE{barndorff-nielsen-shephard:04a,
 AUTHOR = {O. E. Barndorff-Nielsen and N. Shephard},
 YEAR = {2004},
 TITLE = {{Econometric analysis of realized covariation: High frequency based covariance, regression, and correlation in financial economics}},
 JOURNAL = {Econometrica},
 VOLUME = {72},
 NUMBER = {3},
 PAGES = {885--925}
}

@ARTICLE{cai-hu-li-zheng:20a,
 AUTHOR = {T. T. Cai and J. Hu and Y. Li and X. Zheng},
 YEAR = {2020},
 TITLE = {{High-dimensional minimum variance portfolio estimation based on high-frequency data}},
 JOURNAL = {Journal of Econometrics},
 VOLUME = {214},
 NUMBER = {2},
 PAGES = {482--494}
}

@ARTICLE{candes-recht:09a,
 AUTHOR = {E. J. Cand\`{e}s and B. Recht},
 YEAR = {2009},
 TITLE = {{Exact matrix completion via convex optimization}},
 JOURNAL = {Foundations of Computational Mathematics},
 VOLUME = {9},
 NUMBER = {6},
 PAGES = {717--772}
}

@ARTICLE{christensen-podolskij-thamrongrat-veliyev:17a,
 AUTHOR = {K. Christensen and M. Podolskij and N. Thamrongrat and B. Veliyev},
 YEAR = {2017},
 TITLE = {{Inference from high-frequency data: A subsampling approach}},
 JOURNAL = {Journal of Econometrics},
 VOLUME = {197},
 NUMBER = {2},
 PAGES = {245--272}
}

@BOOK{clarke:90a,
 AUTHOR = {F. H. Clarke},
 YEAR = {1990},
 TITLE = {{Optimization and Nonsmooth Analysis}},
 EDITION = {1st},
 PUBLISHER = {Society for Industrial and Applied Mathematics},
 ADDRESS = {Philadelphia}
}

@ARTICLE{delbaen-schachermayer:94a,
 AUTHOR = {F. Delbaen and W. Schachermayer},
 YEAR = {1994},
 TITLE = {{A general version of the fundamental theorem of asset pricing}},
 JOURNAL = {Mathematische Annalen},
 VOLUME = {300},
 NUMBER = {1},
 PAGES = {463--520}
}

@ARTICLE{diop-jacod-todorov:13a,
 AUTHOR = {A. Diop and J. Jacod and V. Todorov},
 YEAR = {2013},
 TITLE = {{Central limit theorems for approximate quadratic variations of pure jump It\^{o} semimartingales}},
 JOURNAL = {Stochastic Processes and their Applications},
 VOLUME = {123},
 NUMBER = {3},
 PAGES = {839--886}
}

@ARTICLE{fissler-podolskij:17a,
 AUTHOR = {T. Fissler and M. Podolskij},
 YEAR = {2017},
 TITLE = {{Testing the maximal rank of the volatility process for continuous diffusions observed with noise}},
 JOURNAL = {Bernoulli},
 VOLUME = {23},
 NUMBER = {4B},
 PAGES = {3021--3066}
}

@ARTICLE{hautsch-kyj-oomen:12a,
 AUTHOR = {N. Hautsch and L. M. Kyj and R. C. A. Oomen},
 YEAR = {2012},
 TITLE = {{A blocking and regularization approach to high dimensional realized covariance estimation}},
 JOURNAL = {Journal of Applied Econometrics},
 VOLUME = {27},
 NUMBER = {4},
 PAGES = {625--645}
}

@MISC{heiny-podolskij:20a,
 AUTHOR = {J. Heiny and M. Podolskij},
 YEAR = {2020},
 TITLE = {{On estimation of quadratic variation for multivariate pure jump semimartingales}},
 NOTE = {preprint arXiv:2009.02786}
}

@ARTICLE{heston:93a,
 AUTHOR = {S. L. Heston},
 YEAR = {1993},
 TITLE = {{A closed-form solution for options with stochastic volatility with applications to bond and currency options}},
 JOURNAL = {Review of Financial Studies},
 VOLUME = {6},
 NUMBER = {2},
 PAGES = {327--343}
}

@ARTICLE{hoerl-kennard:70a,
 AUTHOR = {A. E. Hoerl and R. W. Kennard},
 YEAR = {1970},
 TITLE = {{Ridge regression: Biased estimation for nonorthogonal problems}},
 JOURNAL = {Technometrics},
 VOLUME = {12},
 NUMBER = {1},
 PAGES = {55--67}
}

@TECHREPORT{jacod:94a,
 AUTHOR = {J. Jacod},
 TITLE = {{Limit of random measures associated with the increments of a Brownian semimartingale}},
 INSTITUTION = {Universit\'{e} Pierre et Marie Curie, Paris},
 YEAR = {1994},
 TYPE = {{Preprint number 120, Laboratoire de Probabiliti\'{e}s}}
}

@ARTICLE{jacod:08a,
 AUTHOR = {J. Jacod},
 YEAR = {2008},
 TITLE = {Asymptotic properties of realized power variations and related functionals of semimartingales},
 JOURNAL = {Stochastic Processes and their Applications},
 VOLUME = {118},
 NUMBER = {4},
 PAGES = {517--559}
}

@ARTICLE{jacod-lejay-talay:08a,
 AUTHOR = {J. Jacod and A. Lejay and D. Talay},
 YEAR = {2008},
 TITLE = {{Estimation of the Brownian dimension of a continuous It\^{o} process}},
 JOURNAL = {Bernoulli},
 VOLUME = {14},
 NUMBER = {2},
 PAGES = {469--498}
}

@ARTICLE{jacod-podolskij:13a,
 AUTHOR = {J. Jacod and M. Podolskij},
 YEAR = {2013},
 TITLE = {{A test for the rank of the volatility process: The random perturbation approach}},
 JOURNAL = {Annals of Statistics},
 VOLUME = {41},
 NUMBER = {5},
 PAGES = {2391--2427}
}

@ARTICLE{jacod-podolskij:18a,
 AUTHOR = {J. Jacod and M. Podolskij},
 YEAR = {2018},
 TITLE = {{On the minimal number of driving L\'{e}vy motions in a multivariate price model}},
 JOURNAL = {Journal of Applied Probability},
 VOLUME = {55},
 NUMBER = {3},
 PAGES = {823--833}
}

@BOOK{jacod-protter:12a,
 AUTHOR = {J. Jacod and P. E. Protter},
 YEAR = {2012},
 TITLE = {{Discretization of Processes}},
 EDITION = {2nd},
 PUBLISHER = {Springer},
 ADDRESS = {Berlin}
}

@ARTICLE{kalnina:11a,
 AUTHOR = {I. Kalnina},
 YEAR = {2011},
 TITLE = {{Subsampling high frequency data}},
 JOURNAL = {Journal of Econometrics},
 VOLUME = {161},
 NUMBER = {2},
 PAGES = {262--283}
}

@ARTICLE{koltchinskii-lounici-tsybakov:11a,
 AUTHOR = {V. Koltchinskii and K. Lounici and A. B. Tsybakov},
 YEAR = {2011},
 TITLE = {{Nuclear-norm penalization and optimal rates for noisy low-rank matrix completion}},
 JOURNAL = {Annals of Statistics},
 VOLUME = {39},
 NUMBER = {5},
 PAGES = {2302--2329}
}

@ARTICLE{kong:17a,
 AUTHOR = {X.-B. Kong},
 YEAR = {2017},
 TITLE = {{On the number of common factors underlying large panel high-frequency data}},
 JOURNAL = {Biometrika},
 VOLUME = {104},
 NUMBER = {2},
 PAGES = {397--410}
}

@ARTICLE{kong:20a,
 AUTHOR = {X.-B. Kong},
 YEAR = {2020},
 TITLE = {{A random-perturbation-based rank estimator of the number of factors}},
 JOURNAL = {Biometrika},
 VOLUME = {107},
 NUMBER = {2},
 PAGES = {505--511}
}

@ARTICLE{lounici:14a,
 AUTHOR = {K. Lounici},
 YEAR = {2014},
 TITLE = {{High-dimensional covariance matrix estimation with missing observations}},
 JOURNAL = {Bernoulli},
 VOLUME = {20},
 NUMBER = {3},
 PAGES = {1029--1058}
}

@ARTICLE{lunde-shephard-sheppard:16a,
 AUTHOR = {A. Lunde and N. Shephard and K. Sheppard},
 YEAR = {2016},
 TITLE = {{Econometric analysis of vast covariance matrices using composite realized kernels and their application to portfolio choice}},
 JOURNAL = {Journal of Business and Economic Statistics},
 VOLUME = {34},
 NUMBER = {4},
 PAGES = {504--518}
}

@ARTICLE{mancini:09a,
 AUTHOR = {C. Mancini},
 YEAR = {2009},
 TITLE = {{Non-parametric threshold estimation for models with stochastic diffusion coefficient and jumps}},
 JOURNAL = {Scandinavian Journal of Statistics},
 VOLUME = {36},
 NUMBER = {2},
 PAGES = {270--296}
}

@ARTICLE{marinelli-rockner:16a,
 AUTHOR = {C. Marinelli and M. R\"{o}ckner},
 YEAR = {2016},
 TITLE = {{On the maximal inequalities of Burkholder, Davis and Gundy}},
 JOURNAL = {Expositiones Mathematicae},
 VOLUME = {34},
 NUMBER = {1},
 PAGES = {1--26}
}

@ARTICLE{minsker:17a,
 AUTHOR = {S. Minsker},
 YEAR = {2017},
 TITLE = {{On some extensions of Bernstein's inequality for self-adjoint operators}},
 JOURNAL = {Statistics and Probability Letters},
 VOLUME = {127},
 NUMBER = {1},
 PAGES = {111--119}
}

@ARTICLE{negahban-wainwright:11a,
 AUTHOR = {S. Negahban and M. J. Wainwright},
 YEAR = {2011},
 TITLE = {{Estimation of (near) low-rank matrices with noise and high-dimensional scaling}},
 JOURNAL = {Annals of Statistics},
 VOLUME = {39},
 NUMBER = {2},
 PAGES = {1069--1097}
}

@ARTICLE{pelger:19a,
 AUTHOR = {M. Pelger},
 YEAR = {2019},
 TITLE = {{Large-dimensional factor modeling based on high-frequency observations}},
 JOURNAL = {Journal of Econometrics},
 VOLUME = {208},
 NUMBER = {1},
 PAGES = {23--42}
}

@BOOK{politis-romano-wolf:99a,
 AUTHOR = {D. N. Politis and J. P. Romano and M. Wolf},
 YEAR = {1999},
 TITLE = {{Subsampling}},
 EDITION = {1st},
 PUBLISHER = {Springer},
 ADDRESS = {Berlin}
}

@ARTICLE{recht-fazel-parrilo:10a,
 AUTHOR = {B. Recht and M. Fazel and P. A. Parrilo},
 YEAR = {2010},
 TITLE = {{Guaranteed minimum-rank solutions of linear matrix equations via nuclear norm minimization}},
 JOURNAL = {SIAM Review},
 VOLUME = {52},
 NUMBER = {3},
 PAGES = {471--501}
}

@ARTICLE{reiss-todorov-tauchen:15a,
 AUTHOR = {M. Reiss and V. Todorov and G. Tauchen},
 YEAR = {2015},
 TITLE = {{Nonparametric test for a constant beta between It\^{o} semi-martingales based on high-frequency data}},
 JOURNAL = {Stochastic Processes and their Applications},
 VOLUME = {125},
 NUMBER = {8},
 PAGES = {2955--2988}
}

@ARTICLE{ross:76a,
 AUTHOR = {S. A. Ross},
 YEAR = {1976},
 TITLE = {{The arbitrage theory of capital asset pricing}},
 JOURNAL = {Journal of Economic Theory},
 VOLUME = {13},
 NUMBER = {3},
 PAGES = {341--360}
}

@ARTICLE{seidler-sobukawa:03a,
 AUTHOR = {J. Seidler and T. Sobukawa},
 YEAR = {2003},
 TITLE = {{Exponential integrability of stochastic convolutions}},
 JOURNAL = {Journal of the London Mathematical Society},
 VOLUME = {67},
 NUMBER = {1},
 PAGES = {245--258}
}

@ARTICLE{tibshirani:96a,
 AUTHOR = {R. Tibshirani},
 YEAR = {1996},
 TITLE = {{Regression shrinkage and selection via the lasso}},
 JOURNAL = {Journal of the Royal Statistical Society: Series B},
 VOLUME = {58},
 NUMBER = {1},
 PAGES = {267--288}
}

@ARTICLE{tropp:11a,
 AUTHOR = {J. Tropp},
 YEAR = {2011},
 TITLE = {{Freedman's inequality for matrix martingales}},
 JOURNAL = {Electronic Communications in Probability},
 VOLUME = {16},
 NUMBER = {1},
 PAGES = {262--270}
}

@ARTICLE{tropp:12a,
 AUTHOR = {J. A. Tropp},
 YEAR = {2012},
 TITLE = {{User-friendly tail bounds for sums of random matrices}},
 JOURNAL = {Foundations of Computational Mathematics},
 VOLUME = {12},
 NUMBER = {4},
 PAGES = {389--434}
}

@ARTICLE{tropp:15a,
 AUTHOR = {J. A. Tropp},
 YEAR = {2015},
 TITLE = {{An introduction to matrix concentration inequalities}},
 JOURNAL = {Foundations and Trends{\textregistered} in Machine Learning},
 VOLUME = {8},
 NUMBER = {1--2},
 PAGES = {1--230}
}

@INCOLLECTION{vershynin:10a,
 AUTHOR = {R. Vershynin},
 YEAR = {2010},
 TITLE = {{Introduction to the non-asymptotic analysis of random matrices}},
 BOOKTITLE = {Compressed sensing: Theory and Applications},
 EDITOR = {Y. C. Eldar and G. Kutyniok},
 PAGES = {210--268},
 PUBLISHER = {Cambridge University Press},
 ADDRESS = {Cambridge}
}

@ARTICLE{wang-zou:10a,
 AUTHOR = {Y. Wang and J. Zou},
 YEAR = {2010},
 TITLE = {{Vast volatility matrix estimation for high-frequency financial data}},
 JOURNAL = {Annals of Statistics},
 VOLUME = {38},
 NUMBER = {2},
 PAGES = {943--978}
}

@ARTICLE{watson:92a,
 AUTHOR = {G. A. Watson},
 YEAR = {1992},
 TITLE = {{Characterization of the subdifferential of some matrix norms}},
 JOURNAL = {Linear Algebra and its Applications},
 VOLUME = {170},
 NUMBER = {1},
 PAGES = {33--45}
}

@ARTICLE{zheng-li:11a,
 AUTHOR = {X. Zheng and Y. Li},
 YEAR = {2011},
 TITLE = {{On the estimation of integrated covariance matrices of high dimensional diffusion processes}},
 JOURNAL = {Annals of Statistics},
 VOLUME = {39},
 NUMBER = {6},
 PAGES = {3121--3151}
}

\end{document}